\documentclass[journal]{IEEEtran}
\usepackage[utf8x]{inputenc} 
\usepackage{color}
\usepackage{cite}
\usepackage{amsmath}
\usepackage{amssymb}
\usepackage{paralist}
\usepackage{graphicx}
\usepackage{epstopdf}
\usepackage{epsfig}
\usepackage{bm}
\usepackage{optidef}
\usepackage{amsmath,amscd,latexsym,dsfont,psfrag,pstricks,wrapfig}
\usepackage{tabularx,multirow}
\usepackage[caption=false,font=footnotesize]{subfig}
\usepackage{booktabs}
\usepackage{bm}
\usepackage[ruled,vlined,linesnumbered]{algorithm2e}
\usepackage{algpseudocode}
\usepackage{multicol}
\usepackage{stfloats}
\usepackage{url}
\usepackage{amsfonts,amssymb}
\usepackage[justification=centering]{caption}
\newcommand{\etal}{\textit{et al}.\,}

\usepackage{algpseudocode}
\usepackage{textcomp}

\usepackage{float}

\IEEEaftertitletext{\vspace{-1.5\baselineskip}}

\begin{document}

% theorem 
\newtheorem{theorem}{Theorem}
\newtheorem{proposition}{Proposition}
\newtheorem{lemma}{Lemma}
\newtheorem{definition}{Definition}
\newtheorem{proof}{Proof}

%\setlength{\textfloatsep}{1\baselineskip plus 0.2\baselineskip minus 0.8\baselineskip}

\title{Joint Time and Power Allocation for 5G NR Unlicensed Systems}

\author{Haizhou Bao, Yiming Huo,~\IEEEmembership{Member,~IEEE,} 
Xiaodai Dong,~\IEEEmembership{Senior~Member,~IEEE,}  
and Chuanhe Huang% <-this % stops a space
\thanks{Y. Huo and X. Dong are with the Department of Electrical and Computer
 Engineering, University of Victoria, Victoria, BC V8P 5C2, Canada (e-mail:
ymhuo@uvic.ca, xdong@ece.uvic.ca). H. Bao and C. Huang are with Wuhan University, Wuhan, China  (e-mail: baohzwhu@whu.edu.cn, huangch@whu.edu.cn).
 H. Bao is also a visiting scholar with the Department of Electrical and Computer
 Engineering, University of Victoria, Canada (\emph{Corresponding authors: Xiaodai Dong, Chuanhe Huang}).}}

\maketitle

\begin{abstract}
The fifth-generation (5G) and beyond networks are designed to efficiently utilize the spectrum resources to meet various quality of service (QoS) requirements. The unlicensed frequency bands used by WiFi are mainly deployed for indoor applications and are not always fully occupied. The cellular industry has been working to enable cellular and WiFi coexistence. In particular, 5G New Radio in unlicensed channel spectrum (NR-U) supports the uplink and downlink transmission on the maximum channel occupation time (MCOT) duration. In this paper, we consider maximizing the total throughput of both downlink and uplink in NR-U by jointly optimizing the time and power allocation during MCOT while ensuring fair coexistence with WiFi. Fairness is guaranteed in two steps: 1) tuning the access related parameters of NR-U to achieve proportional fairness, and 2) including 3GPP fairness from the throughput perspective as a constraint in NR-U throughput maximization. %In order to guarantee both the fairness and total throughput maximization of cellular users on unlicensed channels, we first formulate and transform the problem into an equivalent convex problem, and then design a low-complexity algorithm to solve the problem. 
Numerical analysis and simulation 
%conducted based on several spectrum-sharing application scenarios of interest, the proposed algorithm
have demonstrated the superior performance of the proposed resource allocation algorithm compared to conventional deployment strategies.
\end{abstract}

\begin{IEEEkeywords}
5G NR-U, spectrum sharing, resource management, cellular, WiFi, iterative algorithm.
\end{IEEEkeywords}

% \fi
\IEEEpeerreviewmaketitle

% -----------------------------------------------
%                     SECTION  
% ----------------------------------------------- 
\section{Introduction}

\IEEEPARstart{T}he fifth-generation (5G) networks are being fast deployed all over the world with many underlying technologies as critical integral parts. The exponentially growing mobile data traffic and new applications have tremendously pushed the use of the spectrum resource to the limit and hence many new frequency bands, e.g., millimeter wave (mmWave) bands are being adopted for both cellular and WiFi usage \cite{huo/20175g}. Meanwhile, spectrum sharing is another promising approach to address the high demands of data traffic. The 3GPP has been actively seeking to use the unlicensed spectrum since the LTE age. In particular, licensed assisted access (LAA) and LTE-unlicensed (LTE-U) are two protocols proposed to co-exist with WiFi in unlicensed bands. The LTE-U developed in 3GPP Releases 10/11/12 allocates a fraction of a duty cycle for the LTE system and another portion for the WiFi system, which enables the base stations (eNodeB) to offload part of their traffic to the WiFi network \cite{hamidouche/2019contract}. However, there is no carrier sensing before LTE-U transmission, which will degrade the performance of the WiFi system \cite{shoaei2017/efficient}.  On the other hand, the contention-based LAA performs carrier sensing before any transmission.  In \cite{huo/2019enabling}, the authors proposed a cellular-WiFi co-enabled design at the user equipment (UE) end to solve the hardware resource competition issue between the two standards in \cite{shoaei2017/efficient}. %which however may jeopardize the WiFi performance.  
In 5G new radio, the coexistence between NR-unlicensed (NR-U) and WiFi has become a potential technology to boost the throughput of the NR system and improve the quality of service (QoS).

Nevertheless, how to accommodate cellular networks to operate in an unlicensed spectrum and ensure a fair and harmonious coexistence with other unlicensed systems is a challenging problem.  According to the NR-U fairness defined by 3GPP TR 38.899  sub 7 GHz \cite{3GPP/38.889}, assuming two independent networks are deployed in the same area (e.g., NR-U+WiFi and WiFi+WiFi),  the fairness criterion is defined as the NR-U network not degrading the WiFi 802.11n network performance when they are deployed in the area, compared to the case where two WiFi 802.11n networks are deployed, similar to the definition in 3GPP 36.889 \cite{3GPP/36.889}.  %The coexistence evaluation of 
3GPP TR 36.889 \cite{3GPP/36.889}  has provided a paradigm for fairly evaluating the coexistence between two radio access technologies. That is, evaluate two WiFi systems coexisting in a given scenario and then replace one WiFi with LAA for a group of eNBs and UE. Accordingly, there are two types of fairness evaluation methods in the literature. The first type strictly follows the evaluation method  \cite{3GPP/36.889,fang/achieving,gao/unlicensedconf,lagen/2019new,access/PatricielloLBG20}. 
For example,
%%, which firstly considers the WiFi-WiFi system and then uses LTE eNodeB to replace one WiFi AP to construct LTE-WiFi coexistence.  
Gao \etal in  \cite{fang/achieving,gao/unlicensedconf} investigate the fair coexistence between WiFi-WiFi systems, then use eNodeB to replace one WiFi AP, and derive the optimal duty circle and initial contention window size to satisfy the fairness requirement, respectively. 
%Lagen \etal in \cite{lagen/2019new}, \cite{access/PatricielloLBG20} investigate NR-U coexistence by firstly deploying WiFi/WiFi systems of two different operators, and then replacing one WiFi network with an NR-U network to 
%determine the impact of NR-U on WiFi system. 
The second type firstly considers a WiFi-LTE coexistence scenario, uses an equivalent WiFi system to replace the LTE system, and then adjusts some parameters to satisfy the fairness requirement \cite{twc/SunD20}. %For example, Sun \etal firstly consider to obtain the fair coexistence between two WiFi systems, and then use one LTE system to replace one WiFi system \cite{twc/SunD20}. It can adjust the number of nodes to tune the fairness level according to the deployment requirement.  
Furthermore, a virtual WiFi network is created to replace the LTE system and compete with the real WiFi network, where the virtual WiFi network is assumed to obtain the same level of throughput as the LTE system to imitate the impact of eNodeB on a WiFi network in \cite{wang/2018optimal,wang/TVT19}.

Although the aforementioned fairness evaluation scenarios are defined, the fairness metrics are still up for discussion. Normally airtime and throughput are the two metrics for fairness consideration. %However, the two methods about fairness in the previous discussion can not be directly used in the NR-U system, which is different for LTE-U. 
The proportional fairness in LTE-WiFi coexistence in \cite{mobiwac/KeyhanianLLM18,gao2020/achieving,tccn/MehrnoushRSG18} considers the equal airtime and equal throughput per node. However, airtime fairness cannot always guarantee the throughput fairness, as there is a trade-off between throughput and airtime fairness when only adjusting the contention window size \cite{cogsima/Tuladhar0V18}.  The throughput fairness in \cite{fang/achieving,gao/unlicensedconf,twc/SunD20} only considers the successful airtime to satisfy the 3GPP fairness constraint given the physical data rate. However, the physical data rate of the licensed system usually is obtained by adjusting its transmission parameters. 3GPP TR 38.899 \cite{3GPP/38.889} recommends that Category-4 listen-before-talk (LBT) should be adopted for 5G new radio base station (gNB) to access the unlicensed channels.  The authors in \cite{gao2020/achieving} investigated different LBT categories proposed by LTE Release 13 and showed that the proposed LAA LBT cannot always make the WiFi system and LTE system proportionally coexist. Additional operation needs to be taken, such as adjusting the initial contention window size or sending duration of the LTE system. This suggests that parameters in 5G NR-U will also need to be optimized. On the other hand, for the NR-U frame, there are two operation modes of the MCOT, i.e., MCOT with a single or multiple downlink (DL)/uplink (UL) switching point(s) \cite{lagen/2019new}. It was recommended in \cite{3GPP/NR-U-COT}  that the maximum number of DL/UL switching points within one MCOT initiated by gNB should be one to reduce the communication latency, which is different from the LTE-U frame where the MCOT is usually used for downlink transmission or uplink transmission. NR-U adopts the frame structure Type 3, similar to the LTE time division duplex (TDD). How to allocate the time slots in gNB initiated MCOT for uplink and downlink to enable fair coexistence with WiFi and maximize the total throughput of the NR system on the unlicensed channels has not been studied in the literature. 

In this paper, we consider the 5G NR-U coexistence with WiFi. Based on the access procedure of the NR-U and WiFi system, we calculate the throughput for the WiFi systems and the uplink and downlink throughput for NR users. To maximize the total throughput on the unlicensed channels and satisfy the fairness constraint between the two systems, we need to allocate the time and power for the uplink and downlink transmission. The main contributions of this paper are summarized as follows:

\begin{itemize}
\item  We propose a new analytical model for NR-U and WiFi fair coexistence, which involves the NR-U frame structure including DL/UL transmission in COT and access procedures of NR-U and WiFi. Fairness is taken into account in two steps. First, to guarantee the proportional fairness of the two systems in terms of airtime, we derive the optimal initial contention window size for NR-U when it adopts LBT to compete for the unlicensed band. Second, throughput fairness is included as constraints in the subsequent NR-U resource allocation optimization.
\item  To maximize the DL and UL throughput of cellular users on the unlicensed channels while satisfy the throughput fairness constraints, the time and power allocation problem of NR-U is formulated and converted into a convex problem. A low-complexity iterative algorithm is developed to solve this problem.   
%\item  We discuss the feasibility of our proposed method for coexistence between WiGig and NR-U on the 60 GHz mmWave band. 
\end{itemize}

The reaming part of this paper is organized as follows. In Section \ref{sec:related_work}, we survey the coexistence techniques for the NR and WiFi systems. In Section \ref{sec:sys_model}, we model the unlicensed and licensed access probability, the throughput for both systems, the fairness and power constraints, and the uplink and downlink time duration constraints. In Section \ref{sec:pro_formulation}, the throughput maximization of the NR-U gNB is formulated,
 and then converted into a convex problem. In Section \ref{sec:pro_decompose}, we use the Lagrange multiplier to relax the problem and then decompose it into two subproblems, which are solved by the Karush-Kuhn-Tucker (KKT) conditions. In Section \ref{sec:simulation}, simulation is conducted to verify the proposed model and algorithm.  Finally, Section \ref{sec:conclusion} concludes the paper and presents possible future work. %provides the proof the Theorem \ref{the:lilun} and Lemma \ref{lem:tuilun}. 

\section{Related Works}
\label{sec:related_work}
\subsection{Resource Allocation for Unlicensed Spectrum}

LTE coexistence with the WiFi system has been widely researched.  Liu \etal in \cite{liu2019/OpenJournal} investigated comprehensive resource management scenarios in LTE-U systems, which includes single small base stations (SBS), multiple SBSs, device-to-device (D2D) network, vehicular ad hoc network, and unmanned aerial vehicle (UAV) systems. Liu \etal in \cite{LIU/2018ICC} researched user association and resource allocation in unlicensed channels. The unlicensed time slots shared by the WiFi access point (AP) are assumed to be equal to the LTE-U users' to guarantee fairness, which provides airtime fairness for two systems. 
% In \cite{pan2017/TWC}, Pan \etal proposed a distributed resource allocation algorithm for software-defined cellular networks for future 5G networks, the unlicensed channel is used by the sDevices and wDevices for a different time duration, according to the allocated time fraction, which is however not compatible with current contention-based WiFi systems.  
For the 5G NR-U system,
Shi \etal \cite{access/ShiCNF20} investigated the  unlicensed spectrum resource sharing between NR-U and WiFi system, and proposed a distributed channel access mechanism to decide the optimal unlicensed channel for NR-U user offloading traffic.
% Lu \etal in {\color{red}\cite{lagen/2019new} citation number seems not right} proposed a novel mathematical framework capable of modeling an integrated system working at extremely high-frequency, which aggregates licensed and unlicensed mmWave radio access technologies.
In \cite{Song/2019ACCESS}, Song \etal proposed to use the cooperative LBT and $(N+3)$-state semi-Markovian to characterize the effective capacity of NR-U with cooperative communications in unlicensed channels, but they did not consider the improvement of throughput for DL and UL transmission on unlicensed channels due to new NR-U frame.
% -----------------------------------------------

\subsection{Access Control and NR-U Frame on Unlicensed Channels}
The 3GPP Release 16 \cite{3GPP/38.889} points out that the Category-4 LBT should be used to access the unlicensed channels for gNB initiated MCOT. Wang \etal in \cite{Wang/2017WCSP} proposed a network adaptive LAA-LBT strategy which includes a partially-randomized initial clear channel assessment (ICCA) scheme and adaptive-contention-window-size-adjustment scheme, enabling the Category-4 access method.  In \cite{corr/abs-2001-04779}, the authors studied NR-U with Category-4 LBT at gNBs and Category-2 LBT at the UEs. Zheng \etal in \cite{zheng/TVT2020} proposed a 3-D Markov chain to model a LAA Category-4 LBT procedure with a gap period and 3-D Markov chain to model an 802.11e enhanced distributed channel access (EDCA) procedure. Both of them considered transmission priorities, and then derived the normalized throughput and average channel access delay when $N_L$ LAA eNodeBs contended an unlicensed spectrum with $N_W$ WiFi systems. Additionally, Pei \etal in \cite{Pei/2018TVT} derived an explicit expression of the access probability for the Category-4 LBT with both linear and binary exponential backoff mechanism, and fixed contention window size.
% In \cite{bojovic/2019evaluating}, the authors have summarized two kinds of access frames. One is to enable LAA transmission synchronization, i.e., the data transmission beginning at the subframe boundary with the reserved signal, and the other is of the partial frame, i.e., type-3 frame.  

\subsection{Fairness between the Two Radio Access Technologies}
There are several precedent work about the fairness between LTE-U and WiFi systems. Max-min fairness that maximizes the minimum average throughput achievable by users from both LTE-U and WiFi networks is adopted in \cite{globecom/HeSHYZ19,jsac/WangQSL17}. Jain's fairness index is widely used to characterize the throughput-fairness tradeoff \cite{3GPP/NR-U-COT,tvt/TangZCYSH19}. The larger the Jain's fairness index is, the fairer the system is. The maximum Jain's fairness index can be achieved when the two systems have the same throughput. Proportional fairness usually refers to equal node airtime or equal throughput per node between the two systems \cite{mobiwac/KeyhanianLLM18,gao2020/achieving,tccn/MehrnoushRSG18}, which intuitively seems a fair opportunity for both systems to access the unlicensed band.  In \cite{sun2020/towards,gao2020/achieving}, the authors applied the 3GPP fairness constraint in terms of throughput for the WiFi system. %The WiFi system throughput under the coexistence with the LTE system or 5G system should not be less than the WiFi throughput under the impact of another WiFi system which has identical throughput with the cellular system. 
All these methods obtain fairness between WiFi and LTE system through adjusting the access parameters, such as initial backoff window size, the number of sensing slots. 
% {\color{red}
% According to NR-U coexistence fairness defined in 3GPP TR 38.889,  Gao \etal  proposed to achieve the 3GPP fairness by replace the WiFi AP with an LTE eNodeB and derive the optimal duty cycle threshold for LTE-U in \cite{fang/achieving,gao/unlicensedconf}. However, they usually assumed that the throughput fairness is usually denoted as successful transmission airtime, and the data rate of the physical layer in the LTE system is given, which is usually needed to adjust to satisfy some requirements. On the contrary, some authors proposed to replace the LTE system with a additional WiFi system to satisfy 3GPP fairness.  The authors in \cite{wang/2018optimal,wang/TVT19}  have also proposed to replace the eNodeB with an equivalent virtual WiFi AP to simulate the impact of the NR system on WiFi system, which is assumed to be given the same data rate as the eNodeB system obtains. Furthermore, the authors in \cite{sun2020/towards} also proposed to maximize the total LTE-U and WiFi throughput on the unlicensed channel considering 3GPP fairness by replacing LTE-U with a WiFi AP, however, the throughput considered for LTE-U and WiFi is the fraction of successful transmission time and the data rate for LTE system is also given.}
Wang \etal in \cite{wang/TVT19,wang/2018optimal} proposed to maximize the throughput on the unlicensed spectrum and ensure the fairness between SBS and WiFi systems,
% The constraint of fairness is implemented according to the fairness definition of 3GPP, with the payload derived for the 
and proposed a virtual WiFi system to imitate the impact of eNodeB on a WiFi network, which provides a new approach to fairness.

% \subsection{Uplink and Downlink Allocation}
%                     SECTION  
% ----------------------------------------------- 
\section{System Model}
\label{sec:sys_model}
\begin{figure}
\centering
\includegraphics[scale=0.6]{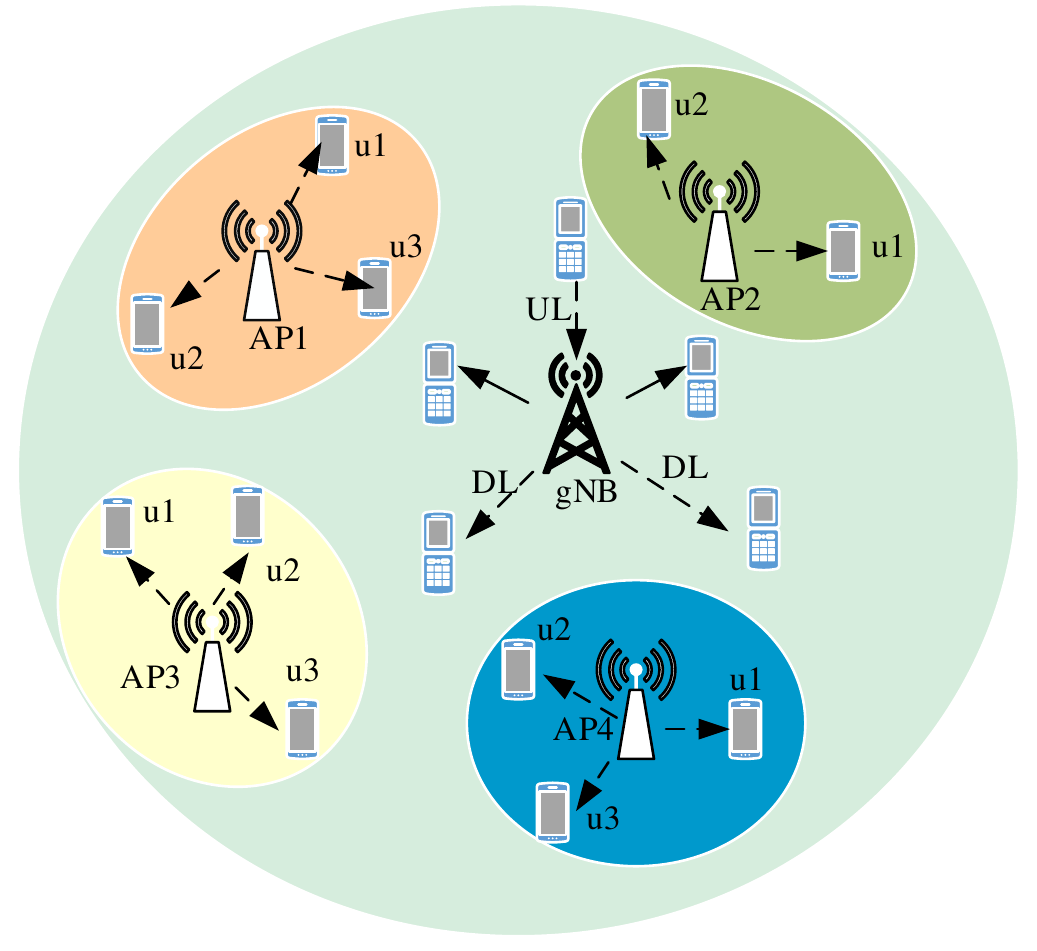}
\caption{Network model.}
\label{fig:network}
\end{figure}   

As shown in Fig.~\ref{fig:network}, we assume that there is a gNB station with $N_u$ cellular user equipment, and there are $K$ WiFi systems. Besides, each WiFi system $k$ consisting of $N_k$ WiFi nodes (including one WiFi AP and $N_k-1$ WiFi stations) utilizes a unique unlicensed channel $f_k$ to avoid interference among $K$ WiFi systems similar to \cite{wang/TVT19}. When the data traffic of some cellular UEs cannot be satisfied, gNB and UEs can offload the data via unlicensed channels to increase the data rate of the UEs. The uplink and downlink UEs can be denoted as $\mathcal{U}$ with $U$ users and $\mathcal{D}$ with $D$ users respectively, and $\mathcal{N}_u={\mathcal{D} \cup \mathcal{U}}$, and $\mathcal{D} \cap \mathcal{U}=\varnothing$, and $|\mathcal{N}_u|=N_u$. We assume that each node can detect other nodes on the same unlicensed carrier with carrier sensing, node buffers are full and there is no hidden terminal.  In our proposed model, we mainly consider the standalone unlicensed band NR-U scenario to align with the 3GPP specification  3GPP TR 38.889 \cite{3GPP/38.889}.
% there are five different deployment scenarios: carrier aggregation between licensed band NR (PCell) and NR-U (SCell); dual connectivity between licensed band LTE (PCell) and NR-U (PSCell); stand-alone NR-U; an NR cell with DL in unlicensed band and UL in licensed band; dual connectivity between licensed band NR (PCell) and NR-U (PSCell). In our proposed model, we mainly consider the standalone unlicensed band NR-U scenario to align with the 3GPP specification. 
Note that the system model includes two stages: airtime competition and joint power and time allocation for transmission. Here, we assume that gNB competes with the $N_k$ WiFi nodes in the unlicensed channel $f_k$ used by the WiFi network $k$. In total there are $K$ unlicensed channels for gNB to coexist with $K$ WiFi networks.  The notations in this paper are specified in Table~\ref{tab:tab1}.

\begin{table*}
\caption{Notation Definition.}
\label{tab:tab1}
\begin{center}
\begin{tabular}{|l|l|l|l|}
  % \hline
  % \multicolumn{2}{|c|}{WiFi System}& \multicolumn{2}{c|}{NR-U System}\\
 \hline
 Parameter &Definition &Parameter &Definition\\
 \hline
 $\mathcal{D}$($\mathcal{U}$) & set of downlink (uplink) users & $\mathbb{E}(PL)$ & mean payload of WiFi \\   
 $m_w (m_l)$ & WiFi (gNB) maximum backoff stage &  $N_k$ & number of WiFi nodes under the coverage of WiFi $AP_k$  \\ 
 $W_{w} (W_l)$ & WiFi (gNB) minimum contention window &  $MCOT$& maximum channel occupation time\\
 $\tau_k^w(\tau_k^l)$ & access probability of WiFi nodes (gNB) &  $P_{tr,w} (P_{tr,l})$ & WiFi (gNB) transmission probability  \\ 
 $p_k^w(p_k^l)$ & collision probability of WiFi nodes (gNB) & $P_{s,w}(P_{s,l})$ & WiFi (gNB) success transmission probability \\ 
 $T_\sigma$ & WiFi slot time/CCA slot time &  $T_k^{s,w} (T_k^{s,l})$& success transmission duration of WiFi (gNB)\\
 $T_d$ & gNB defer time duration & $T_k^{c,w} (T_k^{c,l})$& collision transmission duration of WiFi (gNB)\\
 ACK & Acknowledgment length  &  $T_k^{l,w}$& collision duration between WiFi nodes and gNB\\
 DIFS & distributed interframe space  &  $\overline{T_k^{slot}}$ & total average time duration \\
 SIFS & short interframe space &  $t_{d,k} (t_{u,k})$ & time duration for downlink (uplink) transmission\\
 RTS/CTS & request to send (clear to send) & $p_{d,k} (p_{u,k})$ & transmit power for downlink (uplink) \\
 $\delta$ & propagation delay & $N_u$ & number of licensed users in NR \\
 $\mathbb{E}(PL_{k'})$ & mean payload of virtual WiFi system $k^{'}$ & $T_{gNB}$ & time slot of NR system \\
 \hline
\end{tabular}
\end{center}
\end{table*}

\section{Access Procedure of WiFi and {gNB} and Proportional Fairness}

\subsection{WiFi Access on Unlicensed Channels} 
\label{sec:WiFi}
Under the coverage of the $k^{th}$ WiFi AP, there are $N_{k}$ WiFi nodes to share the same unlicensed channel with gNB, and we assume that standard WiFi is 802.11n and the bandwidth is set to 20 MHz. 
% Based on the distributed coordinated function (DCF) mechanism, 
WiFi nodes will compete with gNB to access the unlicensed channel by adopting an exponential backoff scheme, and here we consider a saturated NR-WiFi coexisting network, i.e., each node in the network always has packets to transmit. Let $\tau_{k}^{w}$ denote the channel access probability of WiFi nodes in WiFi system $k$ in a randomly chosen slot given by \cite{Bianchi/2000JSAC,Hu/2019TVT} 
\begin{equation}
\label{equ:WiFi_trans_proba}
\tau_{k}^{w}=\frac{2(1-2p_{k}^w)}{(1-2p_{k}^w)( W_{w}+1)+p_{k}^wW_{w}(1-(2p_{k}^w)^{m_{w}})},
\end{equation}
where $p_k^w$ is the collision probability for WiFi nodes transmission on the channel $f_k$, $m_{w}$ is the maximum backoff stage, and $W_{w}$ is the minimum contention window size for WiFi nodes.

For the WiFi nodes, the collision occurs when at least one of the remaining $N_{k}-1$ WiFi nodes or gNB access the same unlicensed channel simultaneously with a WiFi user. The collision probability thus is expressed as 
\begin{equation}
\label{equ:WiFi_collison}    
p_{k}^w=1-(1-\tau_{k}^{w})^{N_{k}-1}(1-\tau_{k}^{l}),
\end{equation}
where $\tau_{k}^l$ is the access success probability of gNB, and it is different from $\tau_{k}^w$ because gNB accesses the unlicensed band by adopting different access parameters. Furthermore, the probability that there is at least one WiFi user transmission during a time slot is $P_{tr,w}$; and the probability when only one WiFi user successfully transmits a packet under the condition that at least one WiFi user transmits a packet is $P_{s,w}$, given by
\begin{equation}
\begin{split}
&P_{tr,w}=1-(1-\tau_{k}^w)^{N_{k}},\\
&P_{s,w}=\frac{N_{k}\tau_{k}^w(1-\tau_{k}^w)^{N_{k}-1}}{P_{tr,w}}.
\end{split}
\end{equation}
\subsection{5G NR Access on the Unlicensed Spectrum}

According to the 3GPP \cite{3GPP/38.889}, NR-U enables both uplink and downlink operation in unlicensed channels with multiple switching points or single switching point in the MCOT.
% NR-U supports multiple types of LBT \cite{xiao/2019performance,zheng/TVT2020}, such as Category-1 immediate transmission (no LBT), Category-2  LBT, and Category-4 LBT as summarized in Table \ref{tab:tab2}.
Category-4 LBT channel access can be used for gNB or UE to initiate a COT for normal data transmissions, and it is recommended for the DL and UL switching gap of up to 16 $\mu$s.  To reduce the overhead, 3GPP \cite{3GPP/NR-U-COT} has proposed that the maximum number of DL/UL switching points within a MCOT for a UE-initiated transmission is 2, and the maximum number of switching points for gNB is 1. Once gNB successfully occupies the unlicensed channel $f_k$, it is allowed to use the maximum time duration up to MCOT for transmission.  In this paper, we consider data offloading for NR on unlicensed channels during MCOT initiated by gNB, and correspondingly the switching point during the MCOT is set to 1.  

\begin{figure*}[!tb]
% \begin{tabular}{cc}
\centering
\subfloat[Timing graph of the access procedure for NR and WiFi.]{
% \begin{minipage}[t]{0.48\linewidth}
   \centering
         \includegraphics[scale=0.45]{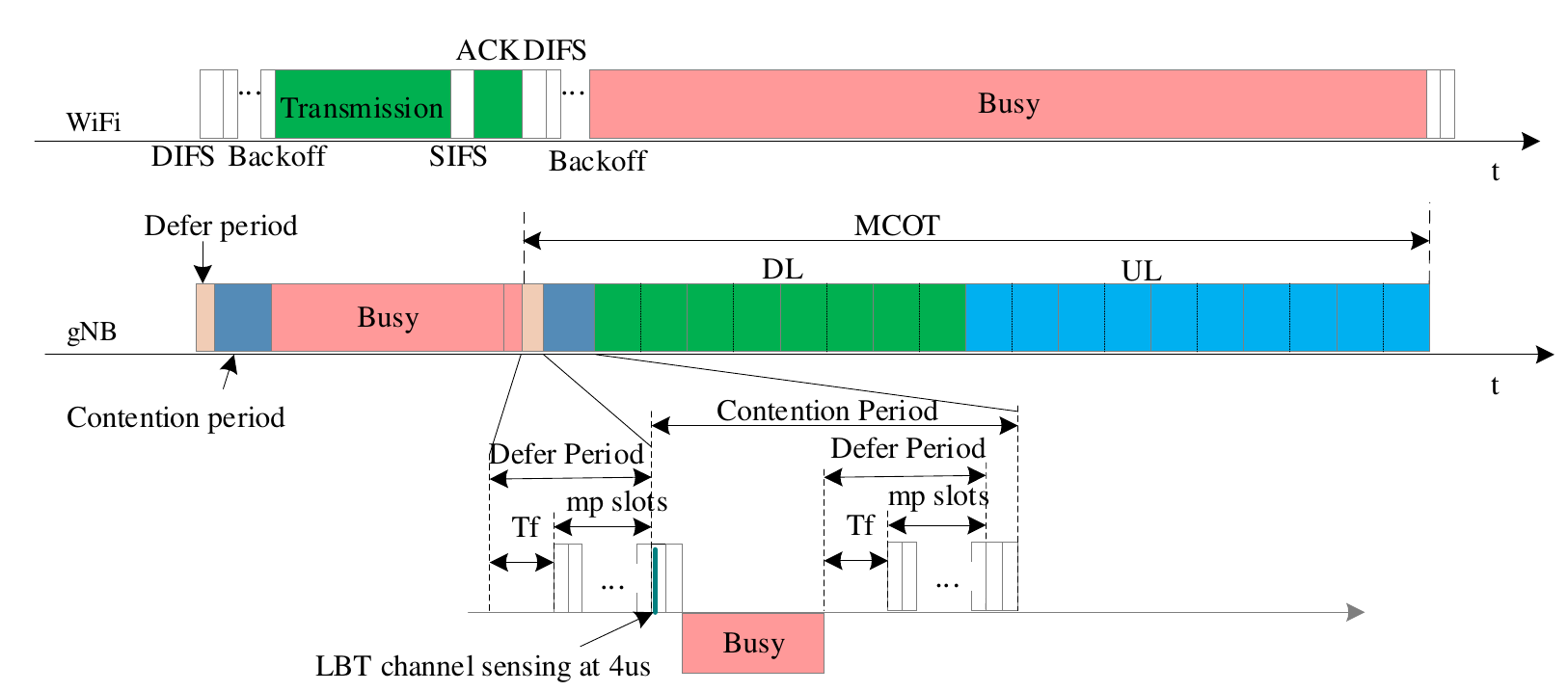}
         % \caption{WiFi $k$ - gNB}
         \label{fig:NR-U-timing}
% \end{minipage}}
}
\subfloat[NR-U DL/UL frames in MCOT.]{
% \begin{minipage}[t]{0.48\linewidth}
      \centering
         \includegraphics[scale=0.45]{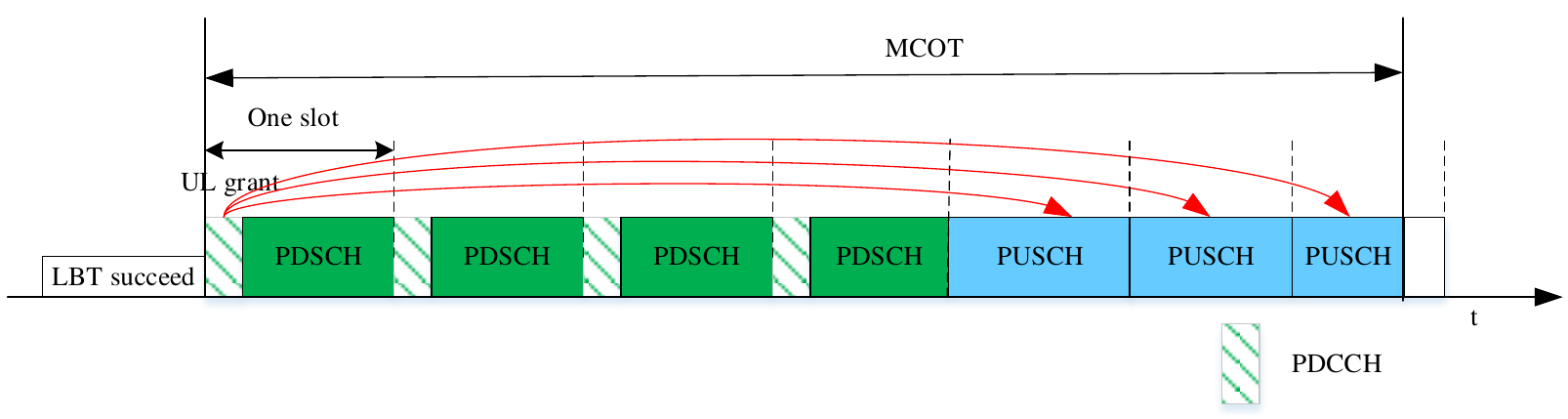}
         % \caption{WiFi $k$ - WiFi $k^{'}$}
        \label{fig:NR-U_frame}
        }
% \end{minipage}}
% \end{tabular}
\caption{Timing graph of access procedure and frame structure for MCOT.}
        \label{fig:MCOT}
\end{figure*}

The procedure for gNB and WiFi competing an unlicensed channel for transmission is depicted in Fig.~(\ref{fig:NR-U-timing}) according to \cite{lien/2016configurable,xiao/2019performance}. WiFi nodes perform the carrier sense multiple access with collision detection (CSMA/CA) channel sense procedure while gNB performs the LBT Category-4 procedure to access the unlicensed channel. If a WiFi node wins the competition after DIFS and backoff procedure, it can transmit the data immediately.
% while gNB will perform the backoff as the channel is busy for gNB at the same time.
After successfully receiving data, the WiFi receiver transmits the ACK message back to the transmitter. After a DIFS period, the WiFi will compete with gNB for the unlicensed channel again.

\begin{table*}[bp]
\caption{Category for NR-U LBT \cite{3GPP/38.889}.}
\label{tab:tab2}
\begin{center}
\begin{tabular}{ |c|c|c|c|c|c| }
\hline
LBT access priority $p$& $CW_{min}$ &$CW_{max}$ & $m_p$ & MCOT& Initial window size\\
\hline
 1 & 3  & 7 & 1 & 2 ms&{3,7}\\ 
 2 & 7  & 15 & 1 & 3 ms& {7,15}\\ 
 3 & 15 & 63 & 3 & 8 ms or 10 ms& {15,31,63}\\ 
 4 & 15 & 1023 &7 & 8 ms or 10 ms& {15,31,63,127,255,511,1023}\\
\hline
\end{tabular}
\end{center}
\end{table*}

On the other hand, the access procedure for gNB is composed of two time periods, the defer period (ICCA stage), the contention period (ECCA stage). The defer period (ICCA) is $T_d=T_f+m_p *T_\sigma$, where $T_f$ is the silent period, $T_\sigma$ is the length of a CCA slot (time slot of gNB system), and the value $m_p$ depends on the LBT access priority $p$ as shown in Table \ref{tab:tab2}. If the channel keeps idle during the ICCA period, gNB will transmit data immediately. Otherwise, gNB will proceed to the ECCA period to compete for the unlicensed channel with WiFi nodes. In the ECCA period, a backoff counter with contention window size $W_i$ is started in each backoff stage, where $W_i \in [0,2^i*W_l-1]$,  $i \in (0,m_l-1)$ is the backoff stage, $m_l$ is the maximum backoff stage, and $W_l$ is the initial contention window size for gNB. If the backoff stage reaches its maximum value, it will drop the packet. Otherwise, it will perform the channel sensing procedure. If the channel is sensed idle during ECCA defer period, the backoff procedure is started, otherwise, it will continue to sense the channel for ECCA defer period until the channel is idle. The backoff counter will decrease by one each time when the channel is sensed idle in a CCA slot. If the backoff counter reaches zero,  gNB will transmit data immediately.  If the channel is sensed busy before the counter reaches zero in a backoff stage, gNB will freeze the backoff counter and continue to sense the channel for ECCA defer period. If the channel is sensed idle during ECCA defer period again, gNB will recover the backoff counter and sense the channel in the next CCA slot. Otherwise, it will sense the channel for another ECCA defer duration until the channel is idle. This sensing procedure will repeat until the data is successfully transmitted or dropped.

The frame structure of NR-U during MCOT is the same as the Type-3 frame defined in 3GPP Release 13, and the frames including DL/UL are shown in  Fig.~(\ref{fig:NR-U_frame}) \cite{3GPP/R1-1912763,3GPP/R1-1913064}. If gNB wins the competition,  gNB can transmit data during COT which is composed of up to 10 subframes, and each subframe usually contains two NR time slots. The MCOT is divided into DL (green square) and UL (blue square) burst, each of which is composed of multiple subframes (or time slots).  Each DL time slot is composed of PDCCH and PDSCH, and the first PDCCH will provide a UL grant (red line) for the subsequent UL transmission. 

Following the work in [26] on LTE and WiFi coexistence, the access probability of gNB on the unlicensed channel of WiFi $k$ network is given by
\begin{equation}
\label{equ:tau_u_k} 
\begin{split}
&\tau_k^l=\frac{2p_k^l(1-2p_k^l)[1+(1-p_k^l)^L]}{[2-2(1-p_k^l)^L](1-3p_k^l+2(p_k^l)^2)+p_k^lH_1},\\
&H_1=(2W_l+1)(1-2p_k^l)+2p_k^lW_l[1-2(2p_k^l)^{(m_l-1)}],
\end{split}
\end{equation} where $p_k^l$ is the conditional collision probability detailed below, $W_l$ is the minimum (i.e., initial) contention window and $m_l$ is the maximum backoff stage, and $L$ is number of CCA time slots in the ICCA duration, i.e., $L=\lfloor \frac{T_d}{T_{\sigma}} \rfloor$ of gNB.

 gNB shares the unlicensed channel of WiFi system $k$ with $N_{k}$ WiFi nodes. In this scenario, the collision probability of gNB competing with at least one WiFi node in the WiFi system $k$ is calculated as 
\begin{equation}
\label{equ:5G_collison}
p_{k}^l=1-(1-\tau_{k}^{w})^{N_{k}}.%there is no interference among a gNB and other gNB 
\end{equation} 

The probability of at least one gNB accessing the unlicensed channel of WiFi system $k$ and the success probability of just one gNB accessing the unlicensed channel under the condition that at least one gNB transmits packets are $P_{tr,l}$, $P_{s,l}$, respectively
\begin{equation}
\begin{split}
&P_{tr,l}=1-(1-\tau_k^l)^{1}=\tau_k^l,\\
&P_{s,l}=\tau_k^l/\tau_k^l=1.
\end{split}
\end{equation}
% where $P_{s,l}$ equal to 1 because we assume there is just one gNB in the system. 

\subsection{Average Time Slot Duration}
% Let $t_{d,k} (t_{u,k})$ denote the time duration that gNB (UE) offloads data for UE (gNB) $d \in \mathcal{D} (u \in \mathcal{U})$ by using the unlicensed channel of WiFi system $k$.
% $t_{u,k}$ denote the time duration that user $u \in \mathcal{U}$ can upload data to gNB by using the unlicensed channel of WiFi system $k$ .
% After that, we can calculate the average time slot duration.
\begin{enumerate}
\item Idle slot: the average duration of the backoff state (idle slot) is calculated as
\begin{equation}\label{equ:Tbf}
\overline{T_{k}^{idle}}=(1-\tau_{k}^l)(1-\tau_{k}^w)^{N_{k}}T_{\sigma},
\end{equation}
where $T_{\sigma}$ is the CCA time slot for the WiFi system and assumed to be the same for gNB.
\item Average time duration for WiFi successful transmission:
the probability of success transmission for a WiFi user in the WiFi system $k$ during a WiFi slot is
$N_{k}\tau_{k}^w (1-\tau_{k}^{w})^{N_{k}-1}(1-\tau_{k}^{l})$. The average success transmission duration of the WiFi AP is 
\begin{equation}\label{equ:Tsw}
\overline{T_{k}^{s,w}}=N_{k}\tau_{k}^w (1-\tau_{k}^{w})^{N_{k}-1}(1-\tau_{k}^{l})T_{k}^{s,w},
\end{equation} 
where $T_{k}^{s,w}$ is the time duration for success transmission of the WiFi user expressed as \cite{Bianchi/2000JSAC}
\begin{equation}
\begin{split}
T_{k}^{s,w}=&(RTS+CTS+H+\mathbb{E}(PL_w)+ACK)/r_w\\
&+3SIFS+DIFS+4\delta,
\end{split}
\end{equation}
where $RTS$,$CTS$ are the size of the handshake messages, $r_w$ is the transmission rate of the WiFi AP. Moreover, $H, \mathbb{E}(PL_{w}), ACK$ stand for the size of the header, the payload of the WiFi frame, and the size of acknowledgment (ACK) message, respectively. $SIFS$ and $DIFS$ represent the shortest inter-frame spacing and time duration of the distributed inter-frame spacing, and $\delta$ is the propagation time. 
\item Average time duration for a gNB success transmission:
the success transmission probability of gNB with $N_{k}$ WiFi nodes is $\tau_{k}^{l}(1-\tau_{k}^{w})^{N_{k}}$. Correspondingly, the average success time duration is
\begin{equation}\label{equ:Tsl}
\overline{T_{k}^{s,l}}=\tau_{k}^{l}(1-\tau_{k}^{w})^{N_{k}}T_{k}^{s,l},
\end{equation}
where $T_{k}^{s,l}=MCOT+T_{gNB}$ is the success transmission time for gNB on the unlicensed channel, and $T_{gNB}$ is the NR slot as a gNB node will not transmit until the beginning of the next gNB slot \cite{gao2020/achieving,Morteza/2018TON}.
% \cite{Wang/2017WCSP,Pei/2018TVT,Morteza/2018TON}.

\item Average time duration for the failed transmission due to the conflict among WiFi nodes: the conflict is caused by at least two WiFi nodes that transmit in the WiFi system $k$ at the same time, and the conflict probability is $(1-(1-\tau_{k}^{w})^{N_{k}}-N_{k}\tau_{k}^{w}(1-\tau_{k}^{w})^{N_{k}-1})(1-\tau_{k}^l)$.
Therefore, the average failure duration for WiFi nodes is given by
\begin{equation}
% \begin{split}
\overline{T_{k}^{c,w}}=(1-(1-\tau_{k}^{w})^{N_{k}}-N_{k}\tau_{k}^{w}(1-\tau_{k}^{w})^{N_{k}-1})(1-\tau_{k}^l)T_{k}^{c,w},
% \end{split}
\end{equation}
where $T_{k}^{c,w}=\frac{RTS}{r_w}+DIFS+\delta$.

\item Average time duration of the failed transmission due to gNB and WiFi nodes:
the conflict probability between a gNB and at least one WiFi user transmitting at the same time is $\tau_{k}^{l}(1-(1-\tau_{k}^{w})^{N_{k}})$, and the average time duration for the  conflict between  gNB and the WiFi nodes is given by 
\begin{equation}
\overline{T_{k}^{l,w}}=\tau_{k}^{l}(1-(1-\tau_{k}^{w})^{N_{k}})T_{k}^{l,w},\\
\end{equation}
where $T_{k}^{l,w}=max(T_{k}^{c,w}, T_{k}^{c,l})$ is the collision duration between the WiFi nodes and  gNB node, and the failed transmission duration for gNB is $T_{k}^{c,l}=MCOT+T_{gNB}$.
% As $T_{k}^{c,w}\simeq 186~ \mu$s with a data rate of 2 Mbps, and $T_{\sigma}=0.1~\mu$s
% while $T_{k}^{s,l}$ and $T_{k}^{c,l}$ are usually larger than $1~ms$. Therefore,we assume $max(T_{k}^{c,w}, T_{k}^{c,l})=T_k^{c,l}$ in this paper.
\end{enumerate}

The average time spent on each state on the unlicensed channel of WiFi system $k$ can thus be calculated as
% \begin{equation}
% \label{equ:T_slot}
% \begin{split}
% \overline{T_{k}^{slot}}&=\overline{T_{k}^{idle}}+\overline{T_{k}^{s,w}}+\overline{T_{k}^{s,l}}+\overline{T_{k}^{c,w}}+\overline{T_k^{l,w}},\\
% % &=(1-\tau_{k}^l)(1-\tau_{k}^w)^{N_{k}}T_{\sigma}\\
% % &+N_{k}\tau_{k}^w (1-\tau_{k}^{w})^{N_{k}-1}(1-\tau_{k}^{l})T_{k}^{s,w}\\
% % &+\tau_{k}^{l}(1-\tau_{k}^{w})^{N_{k}}\sum_{j\in \mathcal{D}}(t_{d,k}+t_{d,k})T_{MCOT}\\
% % &+(1-(1-\tau_{k}^{w})^{N_{k}}-N_{k}\tau_{k}^{w}(1-\tau_{k}^{w})^{N_{k}-1})\\
% % &*(1-\tau_{k}^l)T_{k}^{c,w}+\tau_{k}^{l}(1-(1-\tau_{k}^{w})^{N_{k}})T_{MCOT}\\
% % &=A+\tau_{k}^l(1-\tau_{k}^{w})^{N_{k}}T_{CMOT}\\
% \end{split}
% \end{equation}
\begin{equation}
\label{equ:T_slot}
\begin{split}
\overline{T_{k}^{slot}}&=\overline{T_{k}^{idle}}+\overline{T_{k}^{s,w}}+\overline{T_{k}^{s,l}}+\overline{T_{k}^{c,w}}+\overline{T_k^{l,w}}.\\
\end{split}
\end{equation}

\subsection{Proportional Fairness between the Two Systems}
The Category-4 LBT proposed in 3GPP \cite{3GPP/38.889} enables gNB and WiFi coexistence but it cannot guarantee fair coexistence between the two systems when there exists an aggressive node that occupies most of the unlicensed channel \cite{gao2020/achieving}. To make the two systems fairly coexist on the unlicensed channel, the most intuitive measure is to equalize the successful airtime ratio that the two systems used to successfully transmit on the unlicensed channel. The access parameters of both systems, such as the initial backoff window sizes, the numbers of sensing slots, the maximum backoff stages, and the retry limits and transmission opportunities, can be adjusted to achieve the desired airtime ratio\cite{gao2020/achieving}.   In this paper, we consider adjusting the initial contention window size of the NR system as in 
\cite{cogsima/Tuladhar0V18,gao2020/achieving}.
% \cite{iccchina/Wang0F20,,wang/2017performance}.
The successful airtime ratios for NR and one WiFi node are given by
\begin{equation}
\label{equ:airtimeratio}
\begin{split}
&r_{gNB,s}=\frac{P_{s,l}P_{tr,l}(1-P_{tr,w})T_k^{s,l}}{\overline{T_k^{slot}}} \\
& r_{w,s}=\frac{P_{s,w}P_{tr,w}(1-P_{tr,1})T_k^{s,w}}{\overline{T_k^{slot}}N_k}.
\end{split}
\end{equation}
% where $T_k^{s,l}=MCOT+T_{gNB}$.

Proportional fairness between gNB and WiFi is achieved when each node achieves an identical (or a desired) fraction of time over the unlicensed channel, e.g., $r_{gNB,s}=r_{w,s}$, we can obtain
\begin{equation}
\tau_k^{l}(W_l)=\frac{\tau_k^w T_k^{s,w}}{T_k^{s,l}-\tau_k^w \cdot T_k^{s,l}+\tau_k^w \cdot T_k^{s,w}}. \\
\end{equation}
The optimal initial contention window size of gNB is 
\begin{equation}
 W_l^{*}={\tau_k^{l}}^{-1}(\frac{\tau_k^w T_k^{s,w}}{T_k^{s,l}-\tau_k^w \cdot T_k^{s,l}+\tau_k^w \cdot T_k^{s,w}}),
\end{equation}
where ${\tau_k^{l}}^{-1}$ is the inverse function of $\tau_k^{l}$. In the following context, the optimal initial contention window size $W_l^{*}$ is set as the default value in our proposed method.

\section{Throughput Fairness and Optimal Time and Power Allocation}
% Although the proportional fairness can be satisfied by adjusting the initial contention window size, it cannot usually guarantee the throughput fairness of the two systems.
The throughput and time fairness cannot be obtained simultaneously by adjusting the contention window size alone \cite{cogsima/Tuladhar0V18}. Thus, we first  obtain the throughput of the two systems when they coexist on a single unlicensed band, and then make the two systems satisfy throughput fairness.

\subsection{Data Rate Analysis for WiFi and NR Systems}
\subsubsection{WiFi Data Rate}
For WiFi systems, the normalized system data rate is a ratio of the successful transmission of information packets in a slot and the average duration of a slot time
\cite{Hu/2019TVT,Morteza/2018TON}.
% \cite{Tan/2019TWC,Hu/2019TVT,Morteza/2018TON}.
Therefore, the normalized data rate for the WiFi system $k$ is defined as 
\begin{equation}
\begin{split}
\label{equ:R_WiFi}
% \begin{split}
R_{k}^W&=\frac{P_{tr,w}P_{s,w}(1-P_{tr,l})\mathbb{E}(PL_k)}{\overline{T_{k}^{slot}}}\\
&=\frac{N_{k}\tau_{k}^w(1-\tau_{k}^{w})^{N_{k}-1}(1-\tau_{k}^l)\mathbb{E}(PL_{k})}{\overline{T_{k}^{slot}}}, \forall k \in \mathcal{K}.
 \end{split}
\end{equation}

\subsubsection{NR Data Rate}
The effective downlink data rate for each cellular user by using the unlicensed channel of the WiFi system $k$ is given by
\begin{equation}
 \begin{split}
    R_{d,k}&=\frac{P_{tr,l}P_{s,l}(1-P_{tr,w})t_{d,k}}{\overline{T_{k}^{slot}}}B_k \log_2(1+\frac{p_{d,k}|{h}_{d,k}|^2}{\sigma^2})\\
    &=\frac{{\tau_{k}}^l(1-\tau_{k}^w)^{N_{k}}t_{d,k}}{\overline{T_{k}^{slot}}}B_k\log_2(1+\frac{p_{d,k}|{h}_{d,k}|^2}{\sigma^2}),
 \end{split}
\end{equation}
where $p_{d,k}$ is the transmit power from gNB to cellular UE $d$ by using the WiFi system $k$, ${h}_{d,k}$ is the channel coefficient from gNB to the user $d$,  $\sigma^2$ is the noise power, $t_{d,k}$ is time for gNB downlink transmission to cellular UE, and $B_k$ is the bandwidth.  Similarly, the uplink data rate from user $u \in \mathcal{U}$ to gNB  on the unlicensed channel of WiFi system $k$ is
\begin{equation}
R_{u,k}= \frac{\tau_{k}^l(1-\tau_{k}^w)^{N_{k}} t_{u,k} }{\overline{T_{k}^{slot}}} B_k\log_2(1+\frac{p_{u,k}|{h}_{u,k}|^2}{\sigma^2}),
\end{equation}
where $p_{u,k}$ is the transmit power of cellular UE $u \in \mathcal{U}$ to gNB by using the unlicensed channel $f_k$, $t_{u,k}$ is the time for cellular UE $u \in \mathcal{U}$ uplink transmission, and ${h}_{u,k}$ is the channel coefficient from the cellular user $u$ to gNB. 

Furthermore, the uplink data rate for each user $u \in \mathcal{U}$ and the downlink data rate for each user $d \in \mathcal{D}$ by using the unlicensed channel $f_k$ are calculated as 
\begin{equation}
\begin{split}
&R_k^U=\sum_{u \in \mathcal{U}} R_{u,k}, \forall k \in \mathcal{K},\\
&R_{k}^D=\sum_{d \in \mathcal{D}} R_{d,k}, \forall k \in \mathcal{K}.
\end{split}
\end{equation}

\subsection{Fairness Constraints between Different RATs}

\begin{figure}[!tb]
% \begin{tabular}{cc}
\centering
\subfloat[WiFi $k$ - gNB]{
% \begin{minipage}[t]{0.38\linewidth}
   \centering
         \includegraphics[width=0.24\textwidth]{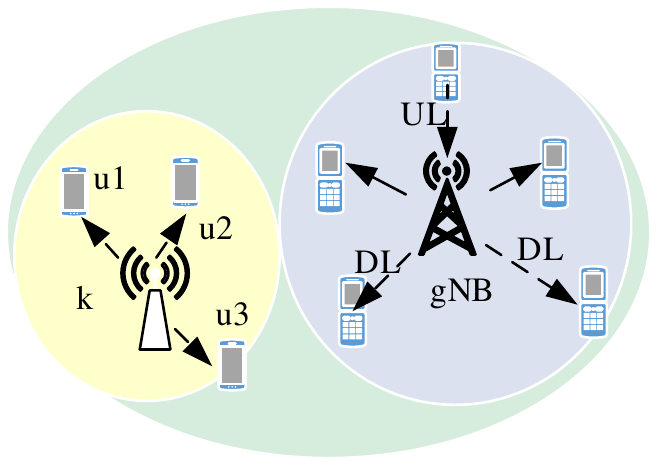}
         % \caption{WiFi $k$ - gNB}
         \label{fig:WiFi_gNB}}
% \end{minipage}}
\subfloat[WiFi $k$ - WiFi $k^{'}$]{
% \begin{minipage}[t]{0.38\linewidth}
      \centering
         \includegraphics[width=0.24\textwidth]{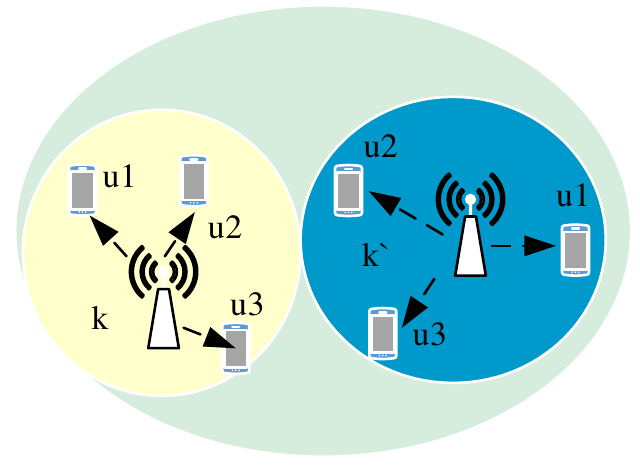}
         % \caption{WiFi $k$ - WiFi $k^{'}$}
         \label{fig:WiFi_WiFi}}
% \end{minipage}}
% \end{tabular}
\caption{Coexistence of WiFi-gNB and WiFi-WiFi.}
        \label{fig:fairness}
\end{figure}

Following the 3GPP fairness definition, Fig.~(\ref{fig:WiFi_gNB}) represents the coexistence between WiFi system $k$ and the NR system, and Fig.~(\ref{fig:WiFi_WiFi}) indicates the coexistence between the WiFi system $k$ and a replaced virtual WiFi system $k^{'}$.  According to the definition to make gNB and WiFi fairly coexist, the data rate of WiFi nodes in the WiFi system $k$ influenced by gNB should not be larger than the data rate of the WiFi system $k$ influenced by a virtual WiFi system $k^{'}$ supporting the same level of traffic load \cite{wang/2018optimal,wang/TVT19,network/HuangCHLR18,infocom/GuanM16}. Firstly, we need to calculate the data rate of the WiFi system $k$ coexisted with the WiFi system $k^{'}$.  To imitate the effect of gNB on WiFi system $k$, WiFi system $k^{'}$ is assumed to have the same data rate and the number of users as the NR system. The two WiFi systems will compete for the unlicensed channel $f_k$, and WiFi system $k^{'}$ is assumed to have  similar parameters as the WiFi system $k$ except the average payload. After replacement, each WiFi system $k$ will just compete with one virtual WiFi system $k^{'}$.  WiFi system $k^{'}$ with $N_u$ nodes can set its payload size to achieve the same data rate as what a gNB can obtain when a gNB coexists with the WiFi system $k$.

\subsubsection{Payload of WiFi System $k'$}

To achieve the same data rate as what a gNB can obtain when a gNB coexists with WiFi system $k$, WiFi system $k'$ with $N_u$ nodes can set its payload size. Note that the data rate of the NR system includes uplink and downlink rates. That is,

\begin{equation}
% \begin{split}
\frac{P_{s}^{k'} P_{tr}^{k'} \mathbb{E}(PL_{k'})}{\left(1-P_{tr}^{k'}\right) \sigma+P_{tr}^{k'} P_{s}^{k'} T_{s}^{k'}+P_{tr}^{k'}\left(1-P_{s}^{k'}\right) T_{c}^{k'}}=R_k^D+R_k^U
% \end{split}
\end{equation}

where
\begin{equation}
\begin{split}
&\tau_{k'}=\frac{2}{1+W_{k}+p_{k'}W_{k}\sum_{t=0}^{m_{w}-1}(2p_{k'})^t},\\
&p_{k'}=1-(1-\tau_{k'})^{N_u-1},\\
&P_{tr}^{k'}=1-(1-\tau_{k'})^{N_u},\\
&P_{s}^{k'}=\frac{N_u\tau_{k'}(1-\tau_{k'})^{N_u-1}}{P_{tr}^{k'}},\\
&T_s^{k'}=\frac{RTS+CTS+(H+\mathbb{E}(PL_{k'})+ACK}{r_w}+3SIFS\\
&+DIFS+4\delta,\\
&T_c^{k'}=\frac{RTS}{r_w}+DIFS+\delta.
\end{split}
\end{equation}
Thus, the average size of the payload of the WiFi system $k^{'}$ is 
\begin{equation}
 \begin{split}
\label{equ:virtual_payload}
&\mathbb{E}(PL_{k'})=\\
&\frac{((1-P_{tr}^{k'})T_{\sigma}+P_{tr}^{k'}(1-P_s^{k'})T_c^{k'}+YP_{tr}^{k'}P_{s}^{k'})r_w (R_{k}^{D}+R_k^U)}{P_{tr}^{k'}P_{s}^{k'}(r_w-(R_{k}^{D}+R_k^U))},
 \end{split}
\end{equation}
where $Y=\frac{RTS+CTS+H+ACK}{r_w}+3SIFS+DIFS+4\delta$.

\subsubsection{WiFi-WiFi System}
The initial backoff window size and the maximum backoff stage of WiFi $k^{'}$ are set to $W_{w}$ and $m_{w}$, respectively, which are identical to the parameters in Section $\ref{sec:WiFi}$. As the two WiFi systems $k$ and $k^{'}$ adopt the same access parameters, the hybrid network can be regarded as a single WiFi network with $N_k+N_u$ WiFi nodes to access the unlicensed channel $f_k$. The data rate of the hybrid network is 
\begin{equation}
\label{equ:R_h}
\begin{split}
&R_{con}=\\
&\frac{P_{tr}^{con} P_s^{con} \mathbb{E}(PL_{con})}{(1-P_{tr}^{con})T_{\sigma}+P_{tr}^{con}P_s^{con} T_{s}^{con}+P_{tr}^{con}(1-P_s^{con})T_c^{con}},
\end{split}
% \caption{data rate of the hybrid Network}
\end{equation}
where
\begin{equation}
\label{equ:para_HY}
\begin{split}
&\tau_{k}^{con}=\frac{2}{1+W_{w}+p_{k}^{con} W_{w} \sum_{t=0}^{m_{w}-1} (2p_{k}^{con})^t},\\
&p_{k}^{con}=1-(1-\tau_{k}^{con})^{N_u+N_k-1},\\
&P_{tr}^{con}=1-(1-\tau_{k}^{con})^{N_u+N_k},\\
&P_s^{con}=\frac{(N_u+N_k)\tau_{k}^{con}(1-\tau_{k}^{con})^{N_k+N_u-1}}{P_{tr}^{con}},\\
&T_s^{con}=\frac{RTS+CTS+(H+\mathbb{E}(PL_{con}))+ACK}{r_w}+3SIFS\\
&+DIFS+4\delta,\\
&T_c^{con}=\frac{RTS}{r_w}+DIFS+\delta,\\
& \mathbb{E}(PL_{con})=\frac{N_k \mathbb{E}(PL_{k})+N_u \mathbb{E}(PL_{k'})}{N_k+N_u}.
\end{split}
\end{equation}
Then, the average data rate that the WiFi system $k$ coexisted with the virtual WiFi system $k^{'}$ is given by equation \eqref{equ:r_w_w}.
\begin{figure*}
\begin{align}
\label{equ:r_w_w}
 % \begin{split}
 R_{k}^{k'}=\frac{N_k \mathbb{E}(PL_k)R_{con}}{N_k \mathbb{E}(PL_k)+N_u \mathbb{E}(PL_{k'})}=\frac{ N_k P_{tr}^{con} P_s^{con} \mathbb{E}(PL_k) }{(N_k+N_u)((1-P_{tr}^{con})T_{\sigma}+P_{tr}^{con}P_s^{con} T_{s}^{con}+P_{tr}^{con}(1-P_s^{con})T_c^{con})}.
 % \end{split}
\end{align}
\end{figure*}

For fair coexistence in term of throughput, the rate relationship must satisfy
\begin{equation}
\label{equ:Fairness}
    R_{k}^{k'} \leq R_{k}^{W}.
\end{equation}

\subsection{DL and UL Time Fraction Constraints}
According to the frame structure of the MCOT, 
% which is composed of multiple DL slots and UL slots.
% Let $t_{d,k}$ and $t_{u,k}$ denote the time duration of user $d$ for DL transmission and user $u$ for UL transmission when the resources of WiFi system $k$ are shared. 
% As can be seen from the frame structure, 
the total downlink time duration $t_{d,k}$ and uplink time duration $t_{u,k}$ for occupying the unlicensed channel of WiFi system $k$ should not be more than MCOT are given by
\begin{equation}
\label{equ:MCOT}
  \sum_{d \in \mathcal{D}}t_{d,k}+\sum_{u \in \mathcal{U}}t_{u,k} \leq MCOT, \forall k \in \mathcal{K}.
\end{equation}

\subsection{Power Constraint}

% {\color{red} Why do the power constraints (26) and (27) include a multiplication of time?\\}
The average uplink power for  each cellular users $u \in \mathcal{U}$ and  total downlink power for all users $d \in \mathcal{D}$ on all unlicensed channels during the MCOT should be smaller than the threshold $P_{avg}$, $P_{gNB}^{max}$, respectively. That is,
\begin{equation}
\label{equ:aver_up_poewr}
    % \sum_{d \in D} t_{d,k}.P_{d,k} \leq P_{k}^{D}, \forall k \in \mathcal{K} 
 \frac{1}{U}\sum_{u \in \mathcal{U}}\sum_{k \in K} \frac{t_{u,k}}{MCOT}. p_{u,k} \leq P_{avg}, \forall u \in \mathcal{U},\\
\end{equation}
\begin{equation}
\label{equ:total_dl_poewr}
 \sum_{d \in \mathcal{D}} \sum_{k \in K}\frac{t_{d,k}}{MCOT}. p_{d,k} \leq P_{gNB}^{max},\\
\end{equation}
where $P_{avg}$ and $P_{d,k}^{max}$ are the average maximum transmit power of each user $u \in \mathcal{U}$  and total maximum transmit power for all users  $d \in \mathcal{D}$ on all unlicensed channels. 

The downlink transmit power for each cellular user $d \in \mathcal{D}$ by gNB in each unlicensed channel $f_k$  by gNB should be smaller than $P_{d,k}^{max}$. That is,
\begin{equation}
\label{equ:dl_poewr}
    % \sum_{d \in D} t_{d,k}.P_{d,k} \leq P_{k}^{D}, \forall k \in \mathcal{K} 
 p_{d,k}\leq P_{d,k}^{max}, \forall d \in \mathcal{D}, k \in \mathcal{K},
\end{equation}
where $P_{d,k}^{max}$ is the maximum downlink transmit power for each user on each unlicensed channel.

\section{NR-U Throughput Maximization Problem Formulation}
\label{sec:pro_formulation}

To improve the unlicensed channel utilization, the problem is formulated to maximize the downlink and uplink throughputs of NR-U on all available unlicensed channels, with both time duration and power constraints considered. That is,
\begin{equation}
\label{equ:P1}
\begin{split}
P1:
\max_{\bm{t_d},\bm{t_u}, \bm{p_d},\bm{p_u}} & \sum_{d \in D}\sum_{k \in \mathcal{K}}R_{d,k}+\sum_{u \in U}\sum_{k \in \mathcal{K}}R_{u,k}\\
\textrm{s.t.}  &~ \eqref{equ:Fairness}, \eqref{equ:MCOT}, \eqref{equ:aver_up_poewr}, \eqref{equ:total_dl_poewr}, \eqref{equ:dl_poewr},\\
              % &~R_{k}^W \geq R_{k}^{k'}, \forall k \in \mathcal{K} \label{subeqn:fairness_p1},\\
              % & \sum_{d \in \mathcal{D}}t_{d,k}+\sum_{u \in \mathcal{U}}t_{u,k} \leq MCOT, \forall k \in \mathcal{K} \label{subeqn:totalTime}, \\
              % & \frac{1}{U}\sum_{u \in \mathcal{U}}\sum_{k \in K} \frac{t_{u,k}}{MCOT}.p_{u,k} \leq P_{avg}, \label{subeqn:uplinkPower}\\
              % & \sum_{d \in \mathcal{D}} \sum_{k \in K} \frac{t_{d,k}}{MCOT}.p_{d,k} \leq P_{gNB}^{max} \label{subeqn:downlinkPower},\\
              % % &\sum_{u \in \mathcal{K}} \frac{\tau_{k}^l(1-\tau_{k}^w)^{N_{k}}t_{d,k}}{\overline{T_{k}^{slot}}} R_{d,k}\geq R_j, \forall j \in \mathcal{D} \\
              % %   & \sum_{u \in \mathcal{K}} \tau_{k}^l(1-\tau_{k}^w)^{N_{k}} t_{d,k}^U R_{d,k} \geq R_j, \forall j \in \mathcal{D}\\
              % % & R_{d}^D \geq R_{d}^{min}, \foall d \in \mathcal{D}\\
              % % & R_{u}^U \geq R_{u}^{min}, \forall u \in \mathcal{U}\\
              % &\sum_{d \in \mathcal{D}}R_{d,k}+\sum_{u \in \mathcal{U}}R_{u,k} < r_w, \forall k \in \mathcal{K}, \label{subeqn:uplink_downlinkPower}\\
              % & p_{d,k}\leq P_{d,k}^{max}\label{subeqn:dl_each_power}, \forall d \in \mathcal{D}, k \in \mathcal{K},\\
              & t_{d,k}, t_{u,k}, p_{u,k}, q_{u,k}\geq 0, \forall d \in \mathcal{D}, u \in \mathcal{U}, k \in \mathcal{K},
              % \label{subeqn:time_variable}.
              % & 0 \leq p_{d,k},\leq P^{g}_{max}\\
              % & 0 \leq p_{u,k} \leq P^{u}_{max}
  \end{split}
\end{equation}
where the variables are $\bm{t_d}=\{t_{d,k}\}_{d\in\mathcal{D},k \in \mathcal{K}}$, $\bm{t_u}=\{t_{u,k}\}_{u\in\mathcal{U},k \in \mathcal{K}}$,  $\bm{p_u}=\{p_{u,k}\}_{u\in\mathcal{U},k \in \mathcal{K}}$, and $\bm{p_d}=\{p_{d,k}\}_{d\in\mathcal{D},k \in \mathcal{K}}$.

To make the problem more tractable, we transform the variables as  $\bm{q}=\{q_{d,k}=p_{d,k}\frac{t_{d,k}}{MCOT}\}_{k\in \mathcal{K}, d \in \mathcal{D}}$, $\overline{\bm{q}}=\{q_{u,k}=p_{u,k}\frac{t_{u,k}}{MCOT}\}_{k\in \mathcal{K}, u \in \mathcal{U}}$, and 
let $\frac{\tau_k^l(1-\tau_k^w)^{N_k}}{\overline{T_{k}^{slot}}}=p_k$. Then the problem is reformulated as 
\begin{subequations}
\label{equ:P2}
\begin{align}
P2:
\max_{\bm{t_d},\bm{t_u}, \bm{q},\bm{\overline{q}}} & \sum_{d \in D}\sum_{k \in \mathcal{K}}R_{d,k}+\sum_{u \in U}\sum_{k \in \mathcal{K}}R_{u,k}\\
\textrm{s.t.} 
              % &~R_{k}^W \geq R_{k}^{k'}, \forall k \in \mathcal{K},\label{subeqn:fairness_p2}\\
              % &\sum_{d \in \mathcal{D}}t_{d,k}+\sum_{u \in \mathcal{U}}t_{u,k} \leq MCOT, \forall k \in \mathcal{K},\\
              & ~~ \eqref{equ:Fairness}, \eqref{equ:MCOT}, \\
              &\frac{1}{U}\sum_{u \in \mathcal{U}}\sum_{k \in K} q_{u,k} \leq P_{avg},\label{subeqn:ue_aver_power}\\
              &	\sum_{d \in \mathcal{D}} \sum_{k \in K} q_{d,k}\leq P_{gNB}^{max},\label{subeqn:dl_total_power}\\
              % & \sum_{d \in \mathcal{D}}R_{d,k}+\sum_{u \in \mathcal{U}}R_{u,k} < r_w,\forall k \in \mathcal{K}, \\
              % &\sum_{u \in \mathcal{K}} \frac{\tau_{k}^l(1-\tau_{k}^w)^{N_{k}}t_{d,k}}{\overline{T_{k}^{slot}}} R_{d,k}\geq R_j, \forall j \in \mathcal{D} \\
              % & \sum_{u \in \mathcal{K}} \tau_{k}^l(1-\tau_{k}^w)^{N_{k}} t_{d,k}^U R_{d,k} \geq R_j, \forall j \in \mathcal{D}\\
              % & \sum_{k \in \mathcal{K}} R_{d,k} \geq R_{d}^{min}, \forall d \in \mathcal{D}\\
              % & \sum_{k \in \mathcal{K}} R_{u,k} \geq R_{u}^{min}, \forall u \in \mathcal{U}\\
              & q_{d,k}\leq \frac{t_{d,k}}{MCOT}.P_{d,k}^{max},\forall d \in \mathcal{D}, k \in \mathcal{K},\label{subeqn:dl_powertime}\\
              & t_{d,k}, t_{u,k}, q_{d,k}, q_{u,k} \geq 0, \forall d \in \mathcal{D}, u \in \mathcal{U}, k \in \mathcal{K},
              %& 0 \leq q_{u,k}\leq t_{u,k}P_{avg}
\end{align}
\end{subequations}

where 
\begin{subequations}
\begin{align}
    &R_{d,k}=p_kt_{d,k}B_k\log_2(1+\frac{MCOT \cdot q_{d,k}|{h}_{d,k}|^2}{\sigma^2 t_{d,k}}),\\
    &R_{u,k}=p_kt_{u,k}B_k\log_2(1+\frac{MCOT \cdot q_{u,k}|{h}_{u,k}|^2}{\sigma^2t_{u,k}}).
    % & p_k=\frac{{\tau_{k}}^l(1-\tau_{k}^w)^{N_{k}}}{\overline{T_{k}^{slot}}}.
\end{align}
\end{subequations}
Equation~(\ref{equ:Fairness}) can be rewritten as 
\begin{equation}
\begin{split}
\label{equ:constriant2}
\frac{R_{k}^D+R_k^U}{r_w-(R_k^D+R_k^U)}&\geq \\
&\frac{r_w Z-r_w(N_k+N_u)(Q+Y)-N_k\mathbb{E}(PL_k)}{N_us_k},
\end{split}
\end{equation}
where $s_k=\frac{((1-P_{tr}^{k'})T_{\sigma}+P_{tr}^{k'}(1-P_s^{k'})T_c^{k'}+YP_{tr}^{k'}P_{s}^{k'})r_w}{P_{tr}^{k'}P_{s}^{k'}}$, $Q=\frac{(1-P_{tr}^{con})T_{\sigma}+P_{tr}^{con}(1-P_s^{con})T_c^{con}}{P_{tr}^{con}P_{s}^{con}}$, and \\$Z=\frac{\overline{T_{k}^{slot}}}{\tau_{k}^w(1-\tau_{k}^{w})^{N_{k}-1}(1-\tau_{k}^l)}$. Afterwards, we transform  (\ref{equ:Fairness}) to
\begin{equation}
\label{equ:Fairness_2}
% \begin{split}
R_{k}^{D}+R_k^U =\sum_{d \in \mathcal{D}}R_{d,k}+\sum_{u \in \mathcal{U}}R_{u,k}
% &=p_k B_k(\sum_{d \in \mathcal{D}}t_{d,k}\log_2(1+\frac{MCOT \cdot q_{d,k}|{h}_{d,k}|^2}{\sigma^2 t_{d,k}})+\sum_{u \in \mathcal{U}} t_{u,k}\log_2(1+\frac{MCOT \cdot q_{u,k}|{h}_{u,k}|^2}{\sigma^2 t_{u,k}}))
\geq  \frac{\phi r_w}{1+\phi},
% \end{split}
\end{equation}
where 
\begin{equation}
  \phi=\frac{r_w Z-r_w(N_k+N_u)(Q+Y)-N_k\mathbb{E}(PL_k)}{N_u s_k}.
\end{equation}

\begin{theorem}
\label{the:lilun}
The objective function of problem (\ref{equ:P2}) is convex with respect to (w.r.t) $  (\bm{t_u,t_d,q_d,q_u}).$
\end{theorem}
\begin{proof}
Please refer to appendix A.
\end{proof}

\begin{lemma}
\label{lem:tuilun}
 The equality of the second constraint should be held.
\end{lemma}
\begin{proof}
Please refer to appendix B.
\end{proof}

\section{Optimal Time and Power Allocation}
\label{sec:pro_decompose}

\subsection{Optimal Time Allocation for Uplink and Downlink Transmission}
According to the basic principle of successive convex approximation (SCA), $P2$ can be divided into two problems, the optimal time allocation problem P3 given power allocation and the optimal power allocation problem P4 given time allocation.  Problem $P3$ is formulated as
\begin{subequations}
\label{equ:P3}
\begin{align}
P3:
\max_{\bm{t_d},\bm{t_u}} & \sum_{d \in \mathcal{D}}\sum_{k \in \mathcal{K}}R_{d,k}+\sum_{u \in U}\sum_{k \in \mathcal{K}}R_{u,k}\\
\textrm{s.t.} 
              & ~ \eqref{equ:Fairness_2}, \eqref{equ:MCOT}, \eqref{subeqn:dl_powertime}\\
                      % &~R_{k}^W \geq R_{k}^{k'}, \forall k \in \mathcal{K}\\
              %&\sum_{d \in \mathcal{D}} R_{d,k} +\sum_{u \in \mathcal{U}}R_{u,k} \geq \frac{\phi r_w}{1+\phi},\forall k \in \mathcal{K},\label{subeqn:fairness_p3}\\
              %&\sum_{d \in \mathcal{D}}t_{d,k}+\sum_{u \in \mathcal{U}}t_{u,k} \leq MCOT, \forall k \in \mathcal{K},\label{subeqn:MCOT_p3}\\
             % & \sum_{d \in \mathcal{D}}R_{d,k}+\sum_{u \in \mathcal{U}}R_{u,k} < r_w,\forall k \in \mathcal{K}, \label{subeqn:throughput_upper} \\
              % &\sum_{k \in K} q_{u,k} \leq P_{avg}, \forall u \in \mathcal{U}\\
              % &\sum_{k \in \mathcal{K}} \sum_{d \in D} q_{d,k}\leq P_{gNB}^{max}\\
              % &\sum_{u \in \mathcal{K}} \frac{\tau_{k}^l(1-\tau_{k}^w)^{N_{k}}t_{d,k}}{\overline{T_{k}^{slot}}} R_{d,k}\geq R_j, \forall j \in \mathcal{D} \\
              % & \sum_{u \in \mathcal{K}} \tau_{k}^l(1-\tau_{k}^w)^{N_{k}} t_{d,k}^U R_{d,k} \geq R_j, \forall j \in \mathcal{D}\\
              % & \sum_{k \in \mathcal{K}} R_{d,k} \geq R_{d}^{min}, \forall d \in \mathcal{D}\\
              % & \sum_{k \in \mathcal{K}} R_{u,k} \geq R_{u}^{min}, \forall u \in \mathcal{U}\\
               %& q_{d,k}\leq \frac{t_{d,k}}{MCOT} \cdot P_{d,k}^{max},\forall d \in \mathcal{D}, k \in \mathcal{K},\\
              & t_{d,k},t_{u,k} \geq 0, \forall d \in \mathcal{D}, u \in \mathcal{U}, k \in \mathcal{K}.
              % & 0 \leq q_{d,k}\leq t_{d,k}P^{g}_{max}
\end{align}
\end{subequations}

The Lagrange function of problem P3 is given by
\begin{equation}
\begin{split}
    &L_1(\bm{t_d},\bm{t_u},\bm{\alpha},\bm{\beta},\bm{\xi},\bm{\epsilon},\bm{\eta})=\sum_{k \in \mathcal{K}} \sum_{d \in \mathcal{D}} R_{d,k} +\sum_{u \in U}\sum_{k \in \mathcal{K}}R_{u,k}\\
    &+\sum_{k \in \mathcal{K}}\alpha_{k} (\sum_{d \in \mathcal{D}}R_{d,k}+\sum_{u \in \mathcal{U}}R_{u,k}-\frac{\phi r_w}{1+\phi})\\
    % &+\sum_{k \in \mathcal{K}}\zeta_{k} (r_w-\sum_{d \in \mathcal{D}}R_{d,k}-\sum_{u \in \mathcal{U}}R_{u,k})
    &+\sum_{k \in \mathcal{K}} \beta_{k}(MCOT- (\sum_{d \in \mathcal{D}}t_{d,k}+\sum_{u \in \mathcal{U}}t_{u,k}))\\
    &+\sum_{k \in \mathcal{K}} \sum_{d \in D} \xi_{d,k}(\frac{t_{d,k}}{MCOT} \cdot P_{d,k}^{max}-q_{d,k})\\
    % &+\gamma(P_{gNB}^{max}-\sum_{d \in \mathcal{D}}\sum_{ k \in \mathcal{K}} q_{d,k})\\
     &+\sum_{u \in \mathcal{U}}\sum_{k \in \mathcal{K}}\epsilon_{u,k} t_{u,k}+\sum_{d \in \mathcal{D}}\sum_{k \in \mathcal{K}}\eta_{d,k} t_{d,k},
    % &+\sum_{q \in \mathcal{D}}\sum_{k \in \mathcal{K}}\mu_{d,k} q_{d,k}\\
    % &+\sum_{d \in \mathcal{D}}\sum_{k \in \mathcal{K}}\xi_{d,k}(t_{d,k}P^{g}_{max}-q_{d,k})
    % &+\sum_{d \in \mathcal{D}} \sigma_d (\sum_{k \in \mathcal{K}} R_{d,k}-R_{d}^{min})\\
    % &=\sum_{k \in \mathcal{K}}\sum_{d \in \mathcal{D}} (1+\alpha_k)R_{d,k} +\sum_{u \in U}\sum_{k \in \mathcal{K}}R_{u,k}\\
    % &-\sum_{k \in \mathcal{K}}\beta_k(\sum_{d \in \mathcal{D}}t_{d,k}+\sum_{u \in \mathcal{U}}t_{u,k})-\gamma \sum_{K \in \mathcal{D}}\sum_{k \in \mathcal{K}} q_{d,k} +\eta\\
    % &+\sum_{j \in \mathcal{D}} \delta_{d} (\sum_{u \in \mathcal{K}} \frac{p_dt_{d,k}^D}{\overline{T_{k}^{slot}}} R_{d,k}^D-R_j^D)\\
    % &+ \sum_{j \in \mathcal{D}} \theta_{d} (\sum_{u \in \mathcal{K}}
    % \tau_{k}^w(1-\tau_{k}^l)^{N_{k}}R_{d,k}^U-R_j^U)\\
    % &=\sum_{j \in \mathcal{D}} (1+\theta_j)\sum_{u \in \mathcal{K}}\frac{p_dt_{d,k}^D}{\overline{T_{k}^{slot}}}R_{d,k}^D\\
    % &+\sum_{j \in \mathcal{D}} (1+\delta_j)\sum_{u \in \mathcal{K}}\frac{p_{k}t_{d,k}^U}{\overline{T_{k}^{slot}}}R_{d,k}^U\\
    % &+\sum_{u \in U}\alpha_{k}(R_{k}^W-R_w^{k'})-\sum_{u\in \mathcal{K}}\beta_{k} \sum_{j \in \mathcal{D}}(t_{d,k}^D+\overline{p_{cat2}}t_{d,k}^U)\\
    % &+\sum_{u \in \mathcal{K}}\beta_{k}(T_{MCOT})-\sum_{j \in \mathcal{D}}(\delta_j R_j^D+\theta_j R_j^U)\\
    % &+\sum_{j \in \mathcal{D}}\sum_{u \in \mathcal{K}} \phi_{d,k} t_{d,k}^D\\
    % &+\sum_{j \in \mathcal{D}}\sum_{u \in \mathcal{K}} \omega_{d,k} t_{d,k}^U
\end{split}
\end{equation}
where $\bm{\alpha}=\{\alpha_k\}_{k \in \mathcal{K}},\bm{\beta}=\{\beta_{k}\}_{k \in \mathcal{K}}, \bm{\xi}={\{\xi_{d,k}\}}_{d \in \mathcal{D},k\in \mathcal{K}}, \bm{\epsilon}=\{\epsilon_{u,k}\}_{u \in \mathcal{U}, k \in \mathcal{K}}, \bm{\eta}=\{\eta_{d,k}\}_{d \in \mathcal{D}, k \in \mathcal{K}}$ are the non-negative Lagrangian multipliers.  The dual problem can be written as 
\begin{equation}
\min_{\bm{\alpha}, \bm{\beta}, \bm{\xi}, \bm{\epsilon},\bm{\eta}}D(\bm{\alpha},\bm{\beta},\bm{\xi},\bm{\epsilon},\bm{\eta}),
\end{equation}
where the dual function of problem (\ref{equ:P1}) is denoted as 
\begin{equation}
    D(\bm{\alpha,\beta,\xi,\epsilon,\eta})=\max_{\bm{t_d},\bm{t_u}}L_1(\bm{t_u},\bm{t_d},\bm{\alpha}, \bm{\beta},\bm{\xi}, \bm{\epsilon},\bm{\eta}).
\end{equation}

According to the KKT conditions,  and the corresponding constraints of $P3$, the optimal solution of problem $P3$ %given $t_{d,k}, t_{u,k}$ 
should satisfy 
\begin{subequations}
\label{equ:Lag_uplink}
\begin{align}
\frac{\partial{L_1}}{\partial{t_{d,k}}}=&(1+\alpha_k)\frac{\partial{R_{d,k}}}{\partial{t_{d,k}}}-\beta_k+\xi_{d,k}\frac{P_{d,k}^{max}}{MCOT}+\eta_{d,k} \notag \\
&=0,\label{subeqn:partial_Rdk}\\
\frac{\partial{L_1}}{\partial{t_{u,k}}}=&(1+\alpha_k)\frac{\partial{R_{u,k}}}{\partial{t_{u,k}}}-\beta_k+\epsilon_{u,k}=0,\label{subeqn:partial_Ruk}
% &\alpha_{k} (\sum_{d \in \mathcal{D}} R_{d,k}+\sum_{u\in \mathcal{U}} R_{u,k}-\frac{\phi r_w}{1+\phi})=0, \label{subeqn:partial_updownlink_lower}\\
% % &\zeta_{k} (r_w-\sum_{d \in \mathcal{D}} R_{d,k}-\sum_{u\in \mathcal{U}} R_{u,k})=0, \label{subeqn:partial_updownlink_upper}\\
% & \xi_{d,k}(\frac{t_{d,k}}{MCOT} \cdot P_{d,k}^{max}-q_{d,k})=0,\\
% & \beta_k(MCOT-(\sum_{d \in \mathcal{D}}t_{d,k}+\sum_{u \in \mathcal{U}}t_{u,k}))= 0,\\
% &\eta_{d,k}t_{d,k}=0,\\
% &\epsilon_{u,k}t_{u,k}=0.
\end{align}
\end{subequations}
where 
\begin{equation}
\begin{split}
 \frac{\partial{R_{d,k}}}{\partial{t_{d,k}}}&=B_kp_k \bigg(\frac{\ln(1+\frac{MCOT \cdot q_{d,k}|{h}_{d,k}|^2}{\sigma^2t_{d,k}})}{\ln2}\\
 &-\frac{MCOT \cdot q_{d,k}|{h}_{d,k}|^2}{(\sigma^2 t_{d,k}+MCOT \cdot q_{d,k}|{h}_{d,k}|^2)\ln2}\bigg).
 \end{split}
\end{equation}
Define function $h(x)$ as
\begin{equation}
h(x)=\frac{\ln(1+x)}{\ln2}-\frac{x}{(\ln2) (1+x)},
\end{equation} 
and the derivative of the uplink and downlink data rate can be written as
\begin{subequations}
\begin{align}
 \frac{\partial{R_{d,k}}}{\partial{t_{d,k}}}&=B_kp_kh(\frac{MCOT \cdot q_{d,k}|{h}_{d,k}|^2}{\sigma^2t_{d,k}})\notag\\
 &=\frac{\beta_k-\xi_{d,k}\frac{P_{d,k}^{max}}{MCOT}+\eta_{d,k}}{1+\alpha_k},\label{subeqn:tdk_derivation}\\
 % &\frac{\partial{R_{d,k}}}{\partial{q_{d,k}}}=\frac{p_k t_{d,k}|{h}_{d,k}|^2}{(\sigma^2t_{d,k}+q_{d,k}|{h}_{d,k}|^2) \ln2}\\
 \frac{\partial{R_{u,k}}}{\partial{t_{u,k}}}&=B_kp_kh(\frac{MCOT \cdot q_{u,k}|{h}_{u,k}|^2}{\sigma^2t_{u,k}})=\frac{\beta_k -\epsilon_{u,k}}{1+\alpha_k}\label{subeqn:tuk_derivation}.
\end{align}
\end{subequations}
It is found that $t_{d,k}=0$ if and only if $q_{d,k}=0$, and $t_{u,k}=0$ if and only if $q_{u,k}=0$.  According to the complementary slackness conditions, $\eta_{d,k}=\epsilon_{u,k}=0$. According to Lemma \ref{lem:tuilun}, it can be found that  $\sum_{d \in \mathcal{D}}t_{d,k}+\sum_{u \in \mathcal{U}}t_{u,k}=MCOT$. 

If $\xi_{d,k}>0$, $q_{d,k}=\frac{t_{d,k}}{MCOT} \cdot P_{d,k}^{max}$, which cannot be equal to the given value, thus $\xi_{d,k}=0$. Therefore, from
(\ref{subeqn:tdk_derivation}) and (\ref{subeqn:tuk_derivation}), we can obtain
\begin{equation}
% t_{u,k}=\frac{q_{u,k}|{h}_{u,k}|^2}{h^{-1}(\frac{\beta_k}{p_k})\sigma^2}\bigg|_{\frac{q_{u,k}}{P^{u}_{max}}}
\label{equ:time_frac_ul_dl}
\frac{\partial{R_{d,k}}}{\partial{t_{d,k}}}=\frac{\partial{R_{u,k}}}{\partial{t_{u,k}}}=\frac{\beta_k}{1+\alpha_k}.
% & t_{u,k}=\frac{q_{u,k}|{h}_{u,k}|^2}{h^{-1}(\frac{\beta_k}{p_k})\sigma^2}\
\end{equation}

% \item If the constraint \eqref{subeqn:fairness_p2} is not active, i.e., $\alpha_k=0$, and $\beta_k>0$. 
According to (\ref{equ:time_frac_ul_dl}), we can use the Lambert W function to denote the uplink time duration and downlink time duration as
\begin{equation}
\begin{split}
\label{equ:time_duration}	
&t_{d,k}=-\frac{MCOT\cdot q_{d,k}|{h}_{d,k}|^2}{\sigma^2 (1+\frac{1}{W(-e^{-(\frac{\beta_k \ln2}{(1+\alpha_k)B_kp_k}+1)})})}\bigg|_{\frac{MCOT \cdot q_{d,k}}{P_{d,k}^{max}}},\\
&t_{u,k}=-\frac{MCOT\cdot q_{u,k}|{h}_{u,k}|^2}{\sigma^2 (1+\frac{1}{W(-e^{-(\frac{\beta_k \ln2}{(1+\alpha_k)B_kp_k}+1)})})}\bigg|_{0},
\end{split}
\end{equation}
where $W(.)$ is Lambert W function, $a|_b=max(a,b)$. According to (\ref{equ:equto1}), $\beta_k$ is the solution of 
\begin{equation}
\begin{split}
\label{equ:time_duration_sum}
&\sum_{d \in \mathcal{D}} -\frac{MCOT\cdot q_{d,k}|{h}_{d,k}|^2}{\sigma^2 (1+\frac{1}{W(-e^{-(\frac{\beta_k \ln2}{B_kp_k)}+1)})})}
\bigg|_{\frac{MCOT \cdot q_{d,k}}{P_{d,k}^{max}}}\\
&+\sum_{u \in \mathcal{U}}-\frac{MCOT \cdot q_{u,k}|{h}_{u,k}|^2}{\sigma^2 (1+\frac{1}{W(-e^{-(\frac{\beta_k \ln2}{B_kp_k}+1)})})} \bigg|_{0}=MCOT,
\end{split}
\end{equation}
which can be solved by the bisection method.

% \item If the constraint \eqref{subeqn:fairness_p2} is active, i.e., $\alpha_k>0$, and $\beta_k>0$. Given $\alpha_k$, thus \eqref{subeqn:fairness_p3} and \eqref{subeqn:MCOT_p3} can be combined to solve the problem.   
% \end{itemize}

Next, we can use the sub-gradient to update the Lagrangian multiplier $\alpha_k$ as
\begin{equation}
\label{equ:alpha_update}
\alpha_k(t+1)=\left[\alpha_k(t)-s_1(t) (\sum_{d \in \mathcal{D}}R_{d,k}+\sum_{u \in \mathcal{U}}R_{u,k}-\frac{\phi r_w}{1+\phi})\right]^{+},
% &\beta_k(t+1)=\left[\beta_k(t)-s_2(t)(1- (\sum_{d \in \mathcal{D}}t_{d,k}+\sum_{u \in \mathcal{U}}t_{u,k}))\right]^{+}\\
% &\gamma(t+1)=\left[\gamma(t)-s_3(t)(P_{gNB}^{max}-\sum_{d \in \mathcal{D}}\sum_{ k \in \mathcal{K}} q_{d,k})\right]^{+}\\
% & \mu_{d,k}(t+1)=\left[\mu_{d,k}(t)-s_4(t)q_{d,k}\right]^{+}\\
\end{equation}
where $s_1(t)$ is the step size and $\left[x\right]^{+} \triangleq max(0,x)$. 

\subsection{Optimal Power Allocation for Uplink and Downlink Transmission} 

Given the time duration for DL and UL allocation ($\bm{t_{u}, t_{d}}$), the problem $P4$ can be reformulated as 
\begin{equation}
\label{equ:P4}
\begin{split}
P4:
\max_{\bm{q_d},\bm{q_u}} & \sum_{d \in D}\sum_{k \in \mathcal{K}}R_{d,k}+\sum_{u \in U}\sum_{k \in \mathcal{K}}R_{u,k}\\
\textrm{s.t.} 
              & ~~\eqref{equ:Fairness_2}, \eqref{subeqn:ue_aver_power}, \eqref{subeqn:dl_total_power},\eqref{subeqn:dl_powertime},\\
              % &\sum_{d \in \mathcal{D}} R_{d,k} +\sum_{u \in \mathcal{U}}R_{u,k} \geq \frac{\phi r_w}{1+\phi}, \forall k \in \mathcal{K},\\
              % &\sum_{d \in \mathcal{D}} R_{d,k} \geq \frac{\phi r_w}{1+\phi},\forall k \in \mathcal{K}\\
              % &\sum_{d \in \mathcal{D}}t_{d,k}+\sum_{u \in \mathcal{U}}t_{u,k} \leq 1, \forall k \in \mathcal{K}\\
              % &\frac{1}{U}\sum_{u \in \mathcal{U}}\sum_{k \in K} q_{u,k} \leq P_{avg},\\
              % &\sum_{k \in \mathcal{K}} \sum_{d \in D} q_{d,k}\leq P_{gNB}^{max},\\
              % & \sum_{d \in \mathcal{D}}R_{d,k}+\sum_{u \in \mathcal{U}}R_{u,k} < r_w,\forall k \in \mathcal{K}, \\
              % &\sum_{k \in K} q_{u,k} \leq P_{avg}, \forall u \in \mathcal{U}\\
              % &\sum_{u \in \mathcal{K}} \frac{\tau_{k}^l(1-\tau_{k}^w)^{N_{k}}t_{d,k}}{\overline{T_{k}^{slot}}} R_{d,k}\geq R_j, \forall j \in \mathcal{D} \\
              % & \sum_{u \in \mathcal{K}} \tau_{k}^l(1-\tau_{k}^w)^{N_{k}} t_{d,k}^U R_{d,k} \geq R_j, \forall j \in \mathcal{D}\\
              % & \sum_{k \in \mathcal{K}} R_{d,k} \geq R_{d}^{min}, \forall d \in \mathcal{D}\\
              % & \sum_{k \in \mathcal{K}} R_{u,k} \geq R_{u}^{min}, \forall u \in \mathcal{U}\\
              % & q_{d,k}\leq \frac{t_{d,k}}{MCOT}.P_{d,k}^{max}, \forall d \in \mathcal{D}, k \in \mathcal{K},\\
              & q_{d,k}, q_{u,k}\geq 0, \forall u \in \mathcal{U}, d \in \mathcal{D}, k \in \mathcal{K}.
\end{split}
\end{equation}
Similarly, the Lagrangian function of (P4) is given by 
\begin{equation}
\begin{split}
 &L_2(\bm{q_d},\bm{q_u}, \bm{\alpha},\bm{\theta}, \bm{\gamma},\bm{\xi},\bm{\psi},\bm{\omega})=\sum_{d \in D}\sum_{k \in \mathcal{K}}R_{d,k}+\sum_{u \in \mathcal{U}}\sum_{k \in \mathcal{K}}R_{u,k}\\
 &+\sum_{k \in \mathcal{K}}\alpha_{k} (\sum_{d \in \mathcal{D}}R_{d,k}+\sum_{u \in U}R_{u,k}-\frac{\phi r_w}{1+\phi})\\
 &+ \theta (P_{avg}-\frac{1}{U}\sum_{u \in \mathcal{U}}\sum_{k \in K}q_{u,k})+\gamma (P_{gNB}^{max}-\sum_{k \in \mathcal{K}} \sum_{d \in D} q_{d,k})\\
 &+\sum_{k \in \mathcal{K}} \sum_{d \in D} \xi_{d,k}(\frac{t_{d,k}}{MCOT}.P_{d,k}^{max}-q_{d,k})\\
 &+ \sum_{k \in \mathcal{K}} \sum_{u \in U} \psi_{u,k} q_{u,k}+\sum_{k \in \mathcal{K}} \sum_{d \in D} \omega_{d,k} q_{d,k}, 
\end{split}
\end{equation}
where $\bm{\alpha}=\{\alpha_{k}\}_{k\in \mathcal{K}}, \bm{\theta}, \bm{\gamma},\bm{\xi}=\{\xi_{d,k}\}_{d \in \mathcal{D}, k\in \mathcal{K}},\bm{\psi}=\{\psi_{u,k}\}_{u \in \mathcal{U}, k\in \mathcal{K}},\bm{\omega}=\{\omega_{d,k}\}_{d \in \mathcal{D}, k\in \mathcal{K}}$ are non-negative Lagrange multiplier. According to the KKT conditions and the complementary slackness conditions, the optimal solution with fixed ($\bm{q_d},\bm{q_u}$) should satisfy
\begin{subequations}
\label{equ:Lag_downlink}
\begin{align}
&\frac{\partial{L_2}}{\partial{q_{d,k}}}=(1+\alpha_k)\frac{\partial{R_{d,k}}}{\partial{q_{d,k}}}-\gamma-\xi_{d,k}+\omega_{d,k}=0 \label{subeqn:partital_dk_complete},\\
&\frac{\partial{L_2}}{\partial{q_{u,k}}}=(1+\alpha_k)\frac{\partial{R_{u,k}}}{\partial{q_{u,k}}}-\frac{\theta}{U}+\psi_{u,k}=0 \label{subeqn:partital_uk_complete},
% & \alpha_{k} (\sum_{d \in \mathcal{D}}R_{d,k}+\sum_{u \in \mathcal{U}}R_{u,k}-\frac{\phi r_w}{1+\phi})=0, \label{subeqn:fairness_4}\\
% &\zeta_{k} (r_w-\sum_{d \in \mathcal{D}} R_{d,k}-\sum_{u\in \mathcal{U}} R_{u,k})=0, \label{subeqn:partial_dluldownlink_upper}\\
% &\gamma (P_{gNB}^{max}-\sum_{k \in \mathcal{K}} \sum_{d \in D} q_{d,k})=0, \label{subeqn:dl_power}\\
% &\theta (P_{avg}-\frac{1}{U}\sum_{u \in \mathcal{U}}\sum_{k \in K} q_{u,k})=0\label{subeqn:ul_power},\\
% &\xi_{d,k}(\frac{t_{d,k}}{MCOT}.P_{d,k}^{max}-q_{d,k})=0\label{subeqn:dl_power_each},\\
% &\psi_{u,k} q_{u,k}=0 \label{subeqn:quk},\\
% & \omega_{d,k} q_{d,k}=0 \label{subeqn:qdk}.
\end{align}
\end{subequations}
where 
\begin{subequations}
\begin{align}
\frac{\partial{R_{d,k}}}{q_{d,k}}&=\frac{MCOT \cdot B_kp_k t_{d,k}|{h}_{d,k}|^2}{(\sigma^2t_{d,k}+MCOT\cdot q_{d,k}|{h}_{d,k}|^2) \ln2} \notag \\
&=\frac{\gamma+\xi_{d,k}-\omega_{d,k}}{1+\alpha_k} \label{subeqn:partital_dk},\\
\frac{\partial{R_{u,k}}}{\partial{q_{u,k}}}&=\frac{MCOT \cdot B_k p_k t_{u,k}|{h}_{u,k}|^2}{(\sigma^2t_{u,k}+MCOT \cdot q_{u,k}|{h}_{u,k}|^2) \ln2}=\frac{\frac{\theta}{U}-\psi_{u,k}}{1+\alpha_k}\label{subeqn:partital_uk}.
\end{align}
\end{subequations}

According to the complementary slackness conditions, we can find that $q_{d,k}=0$ and $q_{u,k}=0$ if and only if $t_{d,k}=0$ and $t_{u,k}=0$ for ($u \in \mathcal{U}, d \in \mathcal{D}, k \in \mathcal{K}$). As we assume $t_{d,k}$ and $t_{u,k}$ are given, $\psi_{u,k}=\omega_{d,k}=0$ can be held. From (\ref{subeqn:partital_uk}),  if $\theta=0$, $(1+\alpha_k)\frac{\partial{R_{u,k}}}{\partial{q_{u,k}}}=0$, that is,  $q_{u,k}=0$ and $t_{u,k}=0$, which makes the fraction $\frac{\partial{R_{u,k}}}{\partial{q_{u,k}}}$ meaningless; According to the complementary slackness conditions, $\theta > 0$ and $P_{avg}-\frac{1}{U}\sum_{u \in \mathcal{U}}\sum_{k \in K}q_{u,k}=0$. Furthermore, $\xi_{d,k}=0$ similar to the previous assumption. 
According to \eqref{subeqn:partital_uk} and $P_{avg}=\frac{1}{U}\sum_{u \in \mathcal{U}}\sum_{k \in K}q_{u,k}=0$, we can obtain 
\begin{equation}
\label{equ:up_power_q'}
\begin{split}
% &q_{u,k}=t_{u,k}(\frac{p_k(1+\theta)}{\gamma \ln2}-\frac{\sigma^2}{|{h}_{d,k}|^2}),\\
% &P_{avg}-\frac{1}{U}\sum_{u \in \mathcal{U}}\sum_{k \in K}q_{u,k}=0,\\
&q_{u,k}=t_{u,k}(\frac{B_kp_k(1+\alpha_k)U}{\theta \ln2}-\frac{\sigma^2}{MCOT \cdot |{h}_{u,k}|^2})\bigg|_{0}, %q_uk>0
\end{split}
\end{equation}
where Lagrangian multiplier $\theta$ is written as
\begin{equation}
\theta=\frac{U \cdot \sum_{k \in \mathcal{K}} \sum_{u \in \mathcal{U}} t_{u,k}B_kp_k(1+\alpha_k)}{\ln2 (\frac{\sigma^2}{MCOT}\sum_{k \in \mathcal{K}} \sum_{u \in \mathcal{U}} t_{u,k}/{h}_{u,k}^2+ U \cdot P_{avg})}.
\end{equation}

According to \eqref{subeqn:partital_dk} and the complementary slackness conditions, if $\gamma+\xi_{d,k}=0$, i.e., $\gamma=\xi_{d,k}=0$, then $\frac{\partial{R_{d,k}}}{\partial{q_{d,k}}}=0$, $t_{d,k}=0$ and $q_{d,k}=0$, which makes $\frac{\partial{R_{d,k}}}{\partial{q_{d,k}}}$ meaningless. If $\gamma+\xi_{d,k}>0$, we can obtain $q_{d,k}$ as follows.

\begin{itemize}

\item If $\xi_{d,k}>0$ and $\gamma \geq 0$, according to (\ref{subeqn:partital_dk}) and $\frac{t_{d,k}}{MCOT}.P_{d,k}^{max}=q_{d,k}$, we have $q_{d,k}$ as
\begin{equation}
\label{equ:xilarger0}
q_{d,k}={t_{d,k}}(\frac{B_kp_k (1+\alpha_k)}{(\gamma+\xi_{d,k}) \ln2}-\frac{\sigma^2}{MCOT \cdot |{h}_{d,k}|^2})\bigg|_{0}^{\frac{t_{d,k}}{MCOT}.P_{d,k}^{max}},
% &\frac{t_{d,k}}{MCOT}.P_{d,k}^{max}=q_{d,k},
% &q_{d,k}=\frac{t_{d,k}}{MCOT} \cdot P_{d,k}^{max}.
% \end{split}
\end{equation}
where the Lagrange multiplier $\gamma+\xi_{d,k}$ is written as
\begin{equation}
\label{equ:rplusesi}
\gamma+\xi_{d,k}= \frac{B_kp_k (1+\alpha_k)}{\ln2 (\frac{\sigma^2}{MCOT \cdot |h_{d,k}^2|}+\frac{P_{d,k}^{max}}{MCOT})},
\end{equation}
and $a\big|_{b}^{c}=min(max(a,b),c)$.

\item If $\xi_{d,k}=0$ and $\gamma>0$, then $\frac{\partial R_{d,k}}{\partial q_{d,k}}=\frac{\gamma}{1+\alpha_k}$ and $P_{gNB}^{max}=\sum_{k \in \mathcal{K}} \sum_{d \in D} q_{d,k}$, and thus we can obtain 
\begin{equation}
\label{equ:gammalarger0}
\begin{split}
 q_{d,k}=t_{d,k}(\frac{B_k p_k (1+\alpha_k)}{\gamma \ln2}-\frac{\sigma^2}{MCOT \cdot |{h}_{d,k}|^2})\bigg|_{0}^{\frac{t_{d,k}}{MCOT}.P_{d,k}^{max}}, \\
 \end{split}
\end{equation}
where Lagrange multiplier $\gamma$ can be written as
\begin{equation}
\label{equ:gamma}
 \gamma=\frac{\sum_{k \in \mathcal{K}}\sum_{d \in \mathcal{D}}t_{d,k}B_k p_k(1+\alpha_k)}{\ln2(\frac{\sigma^2}{MCOT} \sum_{k\in \mathcal{K}}\sum_{d \in \mathcal{D}} t_{d,k}/{h}_{d,k}^2+P_{gNB}^{max})}.\\
 \end{equation}
 \end{itemize}
 
The download power allocation (multiplied by time) $q_{d,k}$ for gNB is summarized as equation \eqref{equ:qdk}.
\begin{figure*}
\begin{align}
% \begin{equation}
\label{equ:qdk}
q_{d,k}=
 \begin{cases}
      t_{d,k}(\frac{B_k p_k (1+\alpha_k)}{(\gamma+\xi_{d,k}) \ln2}-\frac{\sigma^2}{MCOT \cdot |{h}_{d,k}|^2}) \bigg|_{0}^{\frac{t_{d,k}}{MCOT}.P_{d,k}^{max}} \rm{with}~\eqref{equ:rplusesi}, & \xi_{d,k}>0~and~\gamma\geq 0,\\
      t_{d,k}(\frac{B_kp_k (1+\alpha_k)}{\gamma \ln2}-\frac{\sigma^2}{MCOT \cdot |{h}_{d,k}|^2})\bigg|_{0}^{\frac{t_{d,k}}{MCOT}.P_{d,k}^{max}} \rm{with}~\eqref{equ:gamma}, &\xi_{d,k}=0~ and ~\gamma>0.
 \end{cases}
 \end{align}
\end{figure*}

Lastly, the Lagrangian multiplier is updated by sub-gradient as
\begin{equation}
\label{equ:theta_xi_gamma_update}
\begin{split}
&\xi_{d,k}(t+1)=\left[\xi_{d,k}(t)-s_2(t)(\frac{t_{d,k}}{MCOT} \cdot P_{d,k}^{max}-q_{d,k})\right]^{+},\\
&\gamma(t+1)=\left[\gamma(t)-s_3(t)(P_{gNB}^{max}-\sum_{k \in \mathcal{K}} \sum_{d \in D} q_{d,k})\right]^{+},\\
\end{split}
\end{equation}
where $s_i(t), i=1,2,3$, are the step sizes subject to
\begin{equation}
\sum_{t=1}^{\infty} s_i(t)^{2} < \infty \ \mbox{and} \ \sum_{t=1}^{\infty} s_i(t)=\infty, i=1,2,3.
\end{equation}
The overall algorithm is summarized in Algorithm \ref{Alg:Algorithm}.

\begin{algorithm}[!tbp]
\caption{Resource Allocation Algorithm for UL and DL on Unlicensed Channels}  
\label{Alg:Algorithm}
\DontPrintSemicolon
\SetAlgoLined
% \KwIn{$N_k, \mathcal{D}, \mathcal{U}, W_k, m_k, MCOT, T_{\sigma}, RTS,CTS, H, ACK, SIFS,DIFS, r_w,|h_{j,s}|^2$}
% \KwOut{$\bm{t_u},\bm{t_d},\bm{p_d},\bm{p_u}$}
\tcc{Unlicensed Access Procedure}
Initialize $k=1,\bm{\alpha},\gamma, \bm{\xi}$, set $s_i(i=1,2,3)$;\\
\Repeat{ $k > K$}{
Calculate the WiFi access probability $\tau_{k}^{w}$ and gNB access probability $\tau_{k}^{l}$ according to (\ref{equ:WiFi_trans_proba}) and (\ref{equ:tau_u_k}); \\
Calculate the  average time slot $\overline{T_{k}^{slot}}$ and WiFi throughput $R_{k}^{W}$ according to $(\ref{equ:T_slot})$ and (\ref{equ:R_WiFi}), respectively; \\
Calculate the $p_k=\frac{\tau_k^l(1-\tau_k^w)^{N_k}}{\overline{T_{k}^{slot}}}$;\\
\tcc{Resource Allocation on unlicensed channel}
\Repeat{The objective function of (\ref{equ:P2}) converges}{
   With fixed $\{q_{u,k}\}_{u \in \mathcal{U}}, \{q_{d,k}\}_{d \in \mathcal{D}}$, obtain the $\beta_k$ according to (\ref{equ:time_duration_sum}), and then obtain the optimal value $\{t_{u,k}, t_{d,k}\}_{u \in \mathcal{U}, d \in \mathcal{D}}$ according to (\ref{equ:time_duration});\\
   % and obtain $q_{d,k}$ according to (\ref{equ:dl_power_q}) by bisection  method, and calculate the $t_{d,k}$ according to (\ref{equ:dl_time_frac});\\
   Obtain the optimal value $\{q_{u,k}\}_{u \in \mathcal{U}}$ according to ($\ref{equ:up_power_q'}$), and obtain the optimal $\{q_{d,k}\}_{d \in \mathcal{D}}$ according to (\ref{equ:qdk}) with fixed $\{t_{u,k},t_{d,k}\}_{u \in \mathcal{U}, d \in \mathcal{D}}$;\\
    }
   Update the  Lagrangian multiplier $\alpha_k$ as (\ref{equ:alpha_update});\\
   Update the Lagrangian multiplier $\xi_{d,k}$ and $\gamma$ as (\ref{equ:theta_xi_gamma_update});\\
    $k=k+1$;\\
    }
    \end{algorithm}

\section{Performance Analysis, Simulation and Discussion}
\label{sec:simulation}

In this section, we first compare the successful access probability and airtime ratio of our proposed method with the other two methods. Then the performance of the proposed algorithm under the influence of the payload, the maximum downlink power, the length of MCOT, and the number of WiFi nodes is evaluated.

\subsection{System Parameter Setting}

\begin{table*}
\caption{Simulation parameters.}
\label{tab:tab3}
\begin{center}
\begin{tabular}{|l|l|l|l|}
  \hline
  \multicolumn{2}{|c|}{WiFi System (IEEE 802.11n)}& \multicolumn{2}{c|}{NR-U System (sub 7GHz)}\\
  \hline
 Parameter &Value &Parameter &Value\\
  \hline
  % \multirow{13}{4em}{WiFi System}
  H &400 bits &MCOT & 8-10~ms\\
  ACK &  364 bits & U & 5\\
  SIFS &  16 $\mu$s  &D & 5\\
  PIFS & 25 $\mu$s & Subcarrier Spacing (SCS) & 60 KHz\\
  DIFS &  34 $\mu$s  &$T_f$ & 16 $\mu$s \\
  $r_w$ & 54 Mbps  &$T_d$ & $T_f+m_p*T_{\sigma}$\\
  RTS & 288 bits   &$N_u$  &  10\\
  CTS & 352 bits   &$L$   & 8 \\
  $\mathbb{E}(PL_k)$ & 800-2048 Bytes  &$P_{gNB}^{max}$ & {35 dBm}\\
  $T_{\sigma}$ & 9 $\mu$s  &$P_{dk}^{max}$ & 23-35 dBm\\
  $\delta$ & 0.1 $\mu$s  & $P_{avg}$ & 23 dBm\\
  Carrier Bandwidth $B_k$  &20 MHz  & $\bar{d}$ (distance between user and gNB) &10-2000 m\\
  $W_w$& 16 &  $m_l$ & 6 \\
  $m_w$& 6 & $T_{gNB}$ & 0.25 ms \\ 
  $K$& 6 & $N_k$  & 5-30 \\ 
 \hline
\end{tabular}
\end{center}
\end{table*}

The simulation parameters of WiFi and NR-U (sub 7 GHz) are presented in Table~\ref{tab:tab3}. The DL and UL channels in NR-U experience Rayleigh fading and the UMi-street Canyon path loss model is adopted \cite{3GPP/38901} as given by   
\begin{equation}
PL=
\begin{cases}
&32.4+21\log_{10}(\bar{d})+20\log_{10}(f_c), 10~\mbox{m}\leq \bar{d}\leq d_{BP}^{'},\\
&32.4+40\log_{10}(\bar{d})+20\log_{10}(f_c)-9.5\log_{10}((d_{BP}^{'})^2\\
&+(h_{BS}-h_{UT})^2), d_{BP}^{'}<\bar{d}\leq 5~\mbox{km},
\end{cases}
\end{equation}
where $h_{BB}^{'}=h_{BS}-h_E$, $h_{UT}^{'}=h_{UT}-h_E$, $h_E=1~$m, $h_{BS}=10~$m, $h_{UT}=1.5~$m, $f_c=5~$GHz,  
$c=3*10^8~$m/s, $d_{BP}^{'}=4h_{BB}^{'}h_{UT}^{'}f_c/c=300$ m, and $\bar{d}$ is the distance between UE and gNB uniformly chosen from 10 to 2000 m. 
% The pass loss for the uplink and downlink is $128+36.7log_{10}(d) (db)$, where $d$ is the distance between the users and gNB in km, 
The noise power density is $N_0=-174~$dBm/Hz, and the noise power is $-174+10 \log_{10}(B_k)~$dBm.  The simulation is run 100 times, and the result is the average of all runs. Without loss of generality, we assume that there is one WiFi network and hence one unlicensed channel, the number of WiFi nodes varies from 5 to 30, increasing by 5 in each scenario. %in each WiFi network increase by 5, and the minimum and maximum number of WiFi nodes is 5 and 30, respectively.

\subsection{Access Probability and Airtime Ratio}

\begin{figure*}[!tb]
% \begin{tabular}{cc}
\centering
\subfloat[Successful access probability.]{
% \begin{minipage}[t]{0.48\linewidth}
    \centering
     \includegraphics[width=0.38\textwidth]{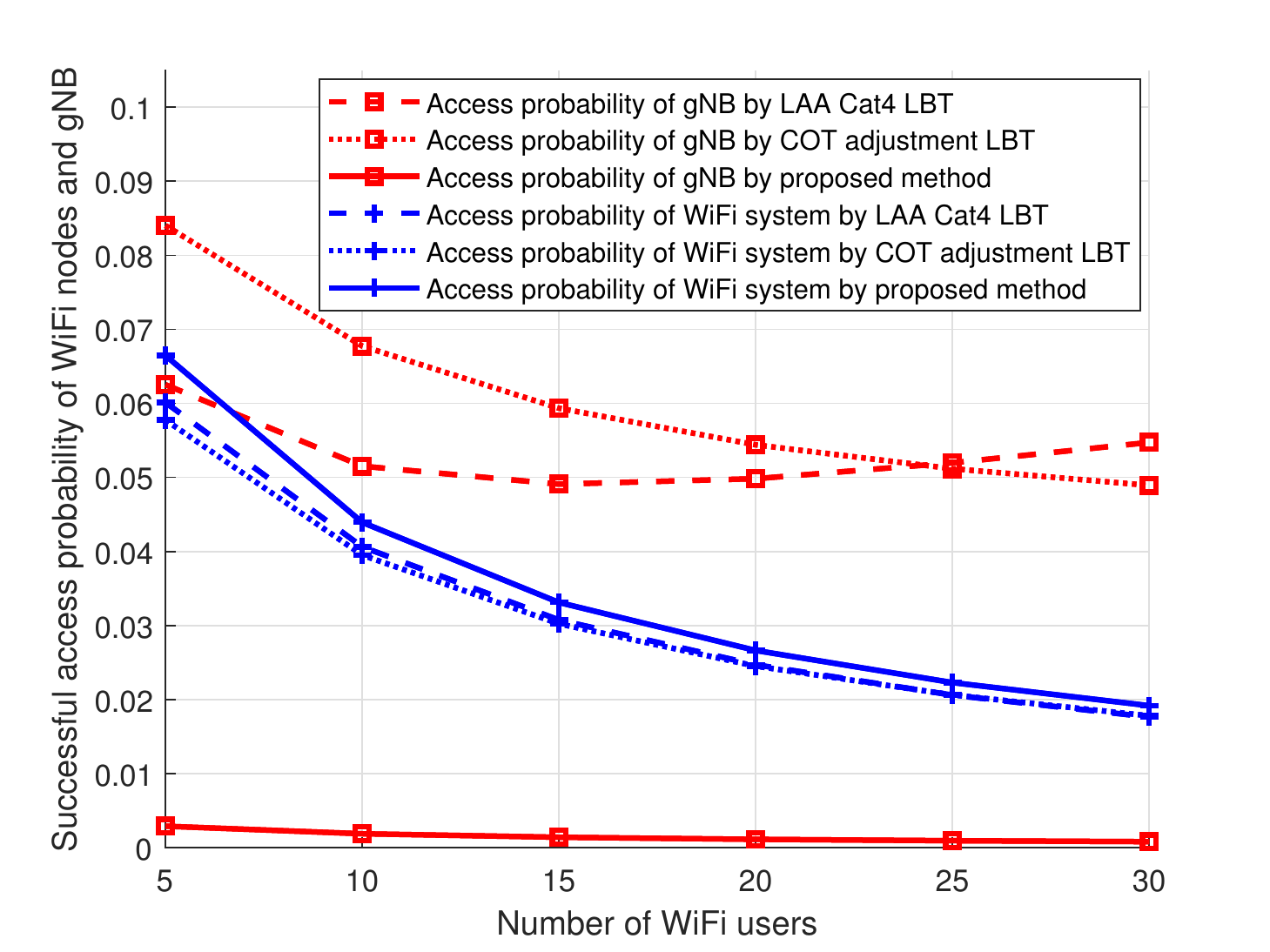}
     % \includegraphics[width=0.38\textwidth]{figureRevised/accessRevised}
    %\caption{Throughput of the NR system on all unlicensed channels.}
    \label{fig:access}}
% \end{minipage}}
\subfloat[Airtime ratio.]{
% \begin{minipage}[t]{0.48\linewidth}
    \includegraphics[width=0.38\textwidth]{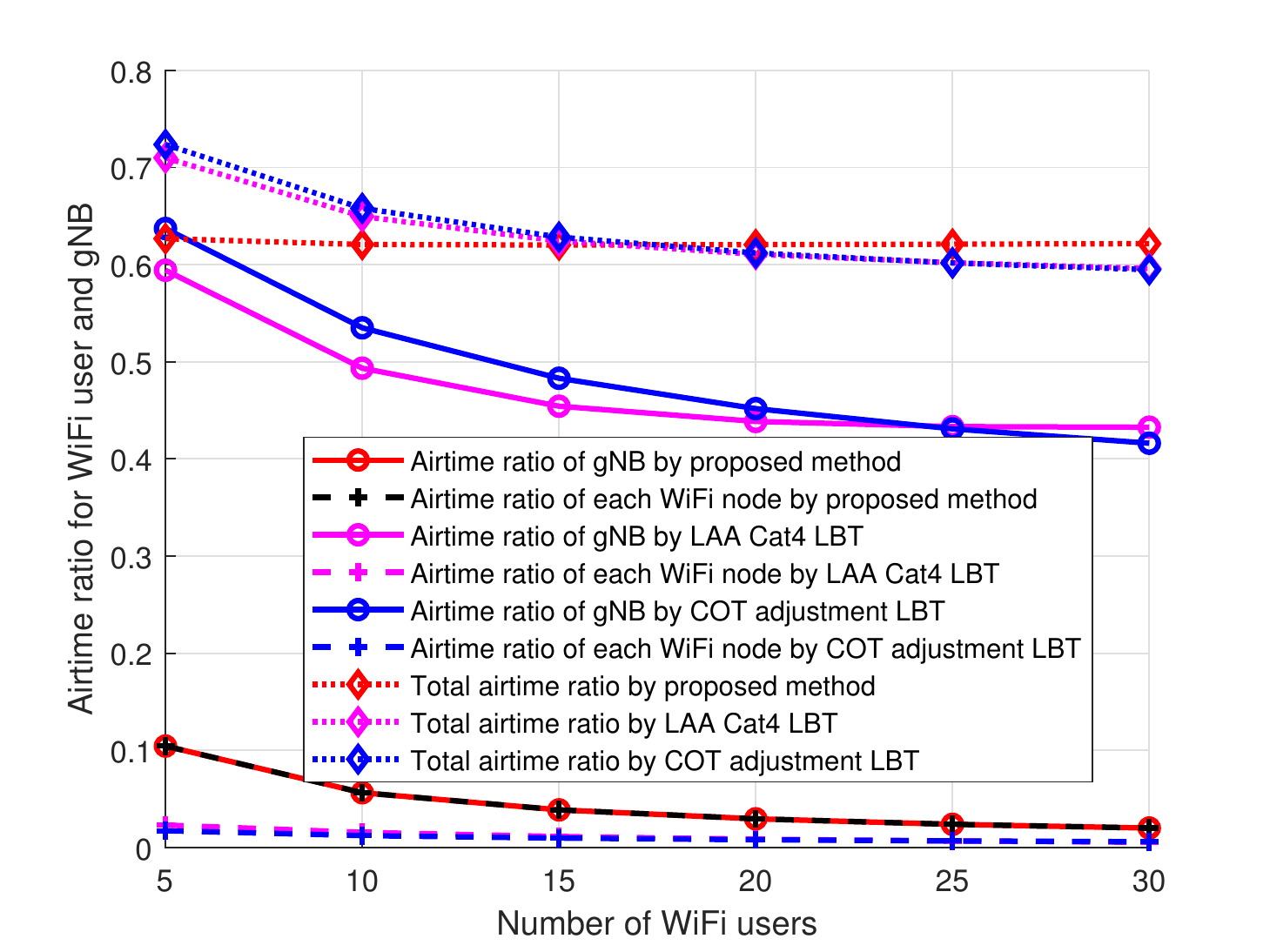}
    % \includegraphics[width=0.38\textwidth]{figureRevised/airtimeRatioRevised}
    % \caption{Licensed access parameters.}
    \label{fig:airtime}}
% \end{minipage}}
% \end{tabular}
\caption{Successful access probability and airtime ratio for NR and WiFi system.} 
\label{fig:main1}
\end{figure*}

In this simulation, we evaluate and compare Cat4 LBT in \cite{Pei/2018TVT} (Category-4 LBT with the initial contention window size 16), and COT adjustment LBT \cite{wang/TVT19} with our proposed method in terms of successful access probability, airtime ratio  as shown in Fig.~(\ref{fig:access}) and Fig.~(\ref{fig:airtime}). In Fig.~(\ref{fig:access}), we can find that almost in all methods the successful access probability of gNB and WiFi decreases with the increase of the number of WiFi nodes, as expected. Because when more WiFi nodes compete for unlicensed channels, the successful access probability for each node will decrease. It is also observed that the access probability of gNB in our proposed setting is smaller than that of Cat4 LBT  and COT adjustment LBT, while the successful access probability for WiFi is slightly improved than the other two methods. The reason gNB has a much smaller access probability than WiFi is that gNB usually has a much larger channel occupation time than WiFi and to ensure an equal airtime ratio per WiFi node to that of gNB, the optimized initial window size of gNB tends to be large to protect the WiFi system.

In Fig.~(\ref{fig:airtime}), we compared the airtime ratio per node (successful transmission time ratio) with different methods. Our proposed method can achieve the equal airtime ratio per node when gNB adopts the optimal initial contention window $W_l^{*}$, which gives the two systems an equal chance to access the unlicensed channel. Note that when the number of WiFi nodes is large, the airtime ratio per WiFi node (lower) and the airtime ratio of the gNB (higher) can be tuned to balance the coexistence, achieving proportional fairness. When gNB adopts the Cat4 LBT and COT adjustment LBT, we find that the airtime ratio of gNB is larger than that of each WiFi node for these two methods, which make the WiFi nodes have little chance to access the unlicensed channel.  On the other hand, the total airtime ratio for all nodes of Cat4 LBT and  COT adjustment LBT decrease with the increase of the number of WiFi nodes, while the proposed method keeps almost constant since the proposed method can adjust the initial contention window size according to the number of WiFi nodes to keep the total airtime ratio constant but the airtime ratio of each node will decrease with the number of WiFi nodes.

\begin{figure*}[!tb]
\centering
\subfloat[Fairness with different downlink power.]{
    \includegraphics[width=0.35\textwidth]{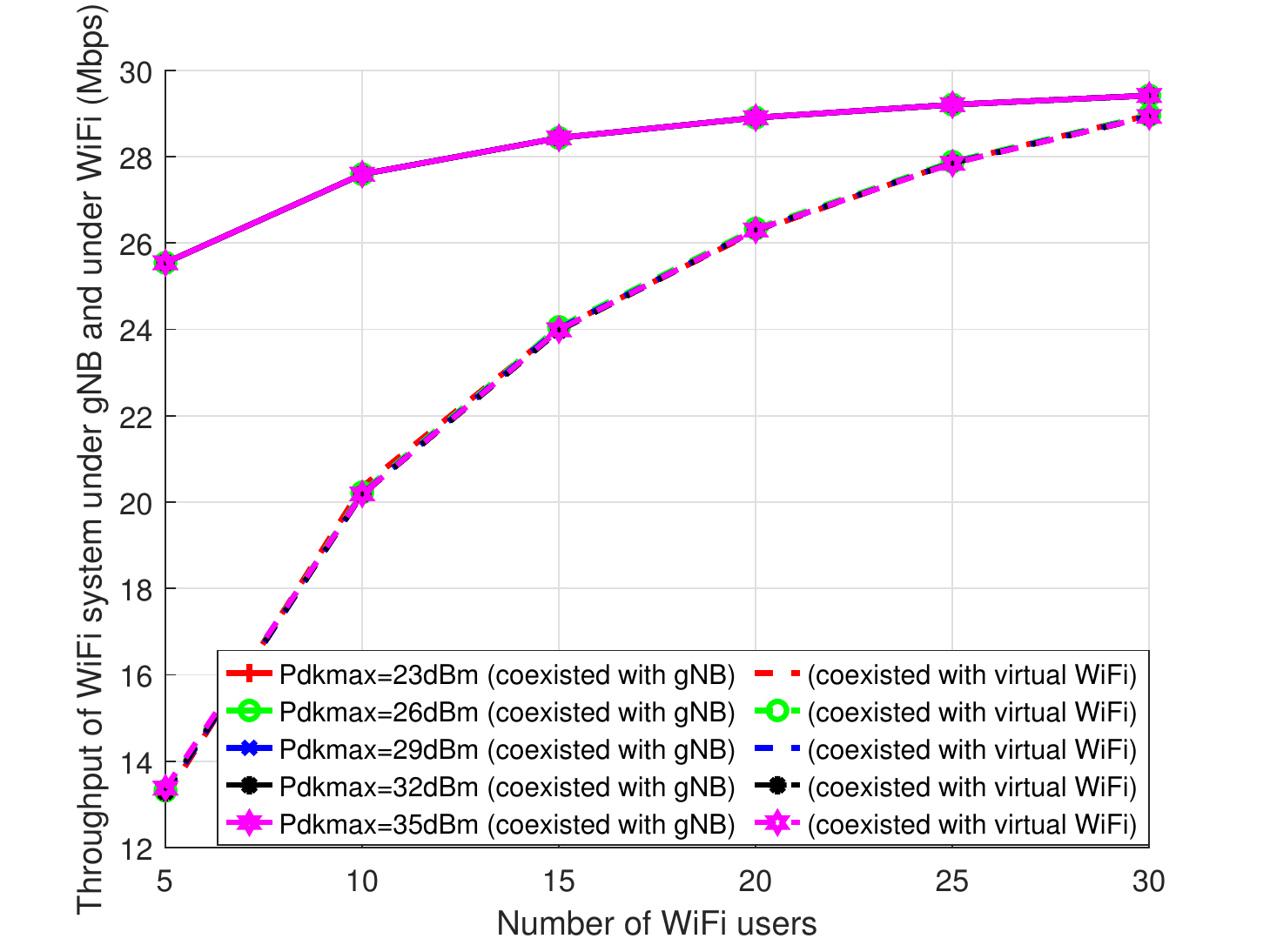}
    \label{fig:fair_power}
    }
\subfloat[Fairness with different WiFi payloads.]{
    \includegraphics[width=0.35\textwidth]{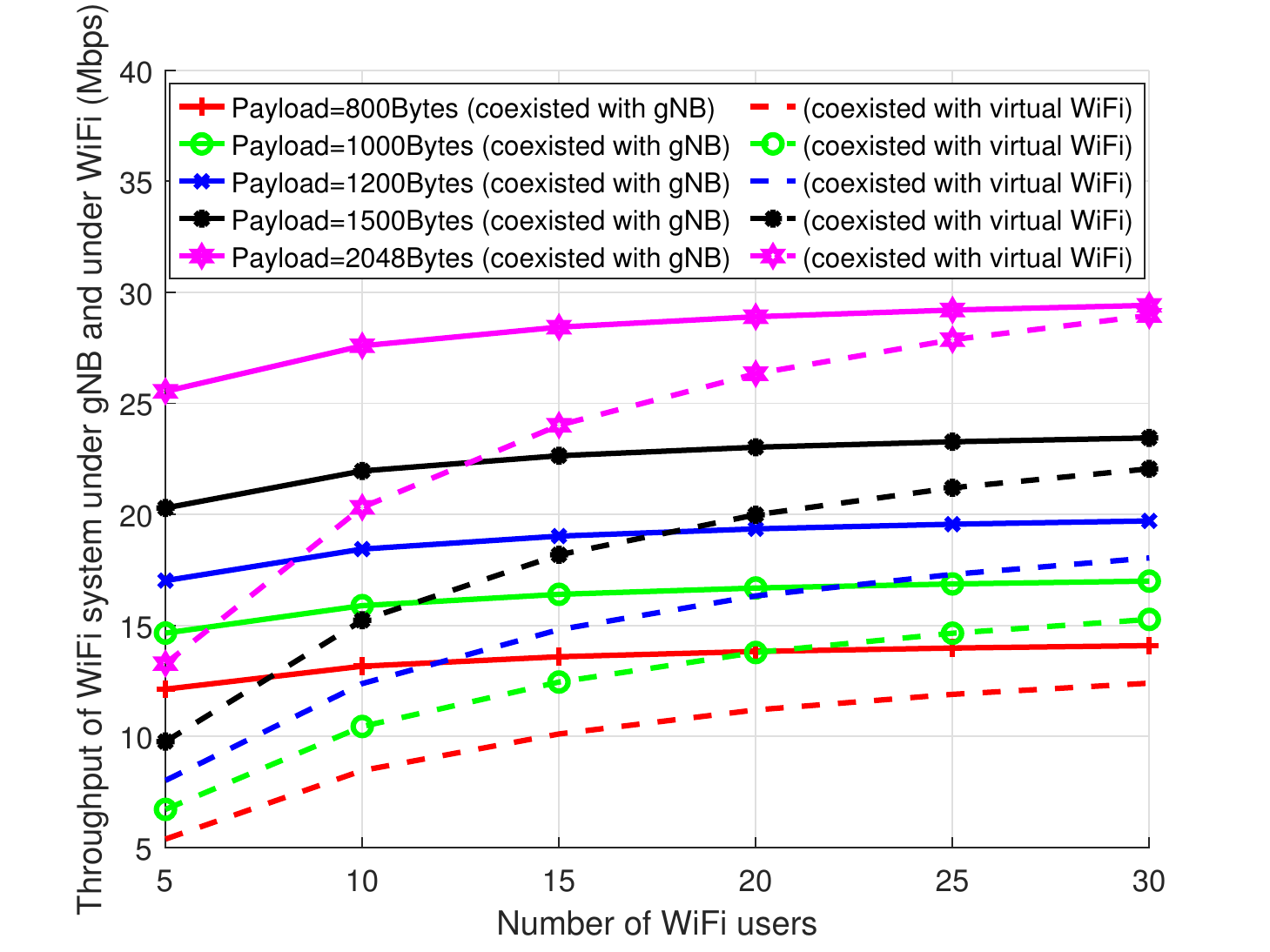}
    \label{fig:fair_payload}
}
\\
\subfloat[Fairness with different MCOTs.]{
    \includegraphics[width=0.35\textwidth]{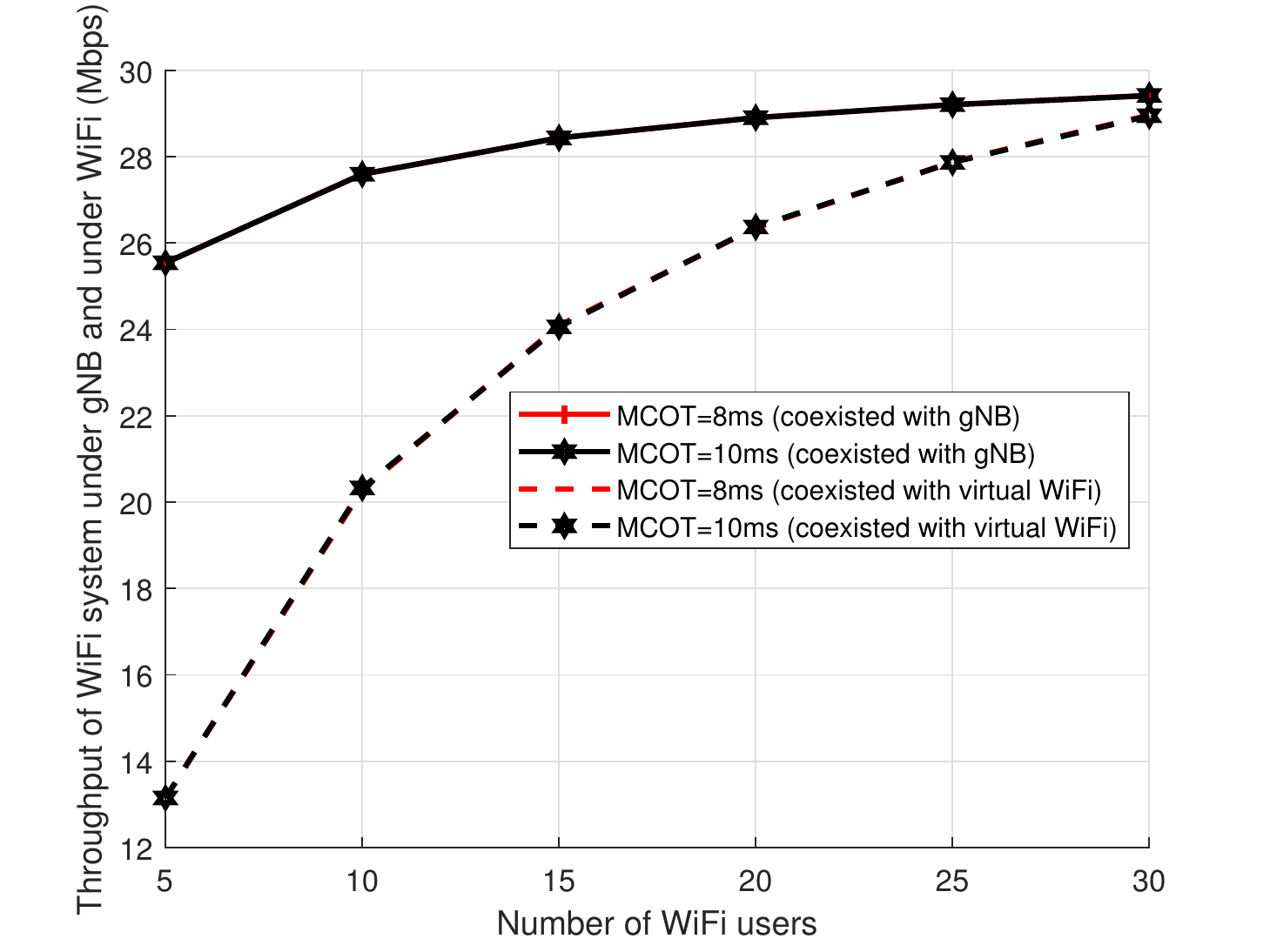}
    \label{fig:fair_mcot}
}
\subfloat[WiFi throughput with different access methods.]{
    \includegraphics[width=0.35\textwidth]{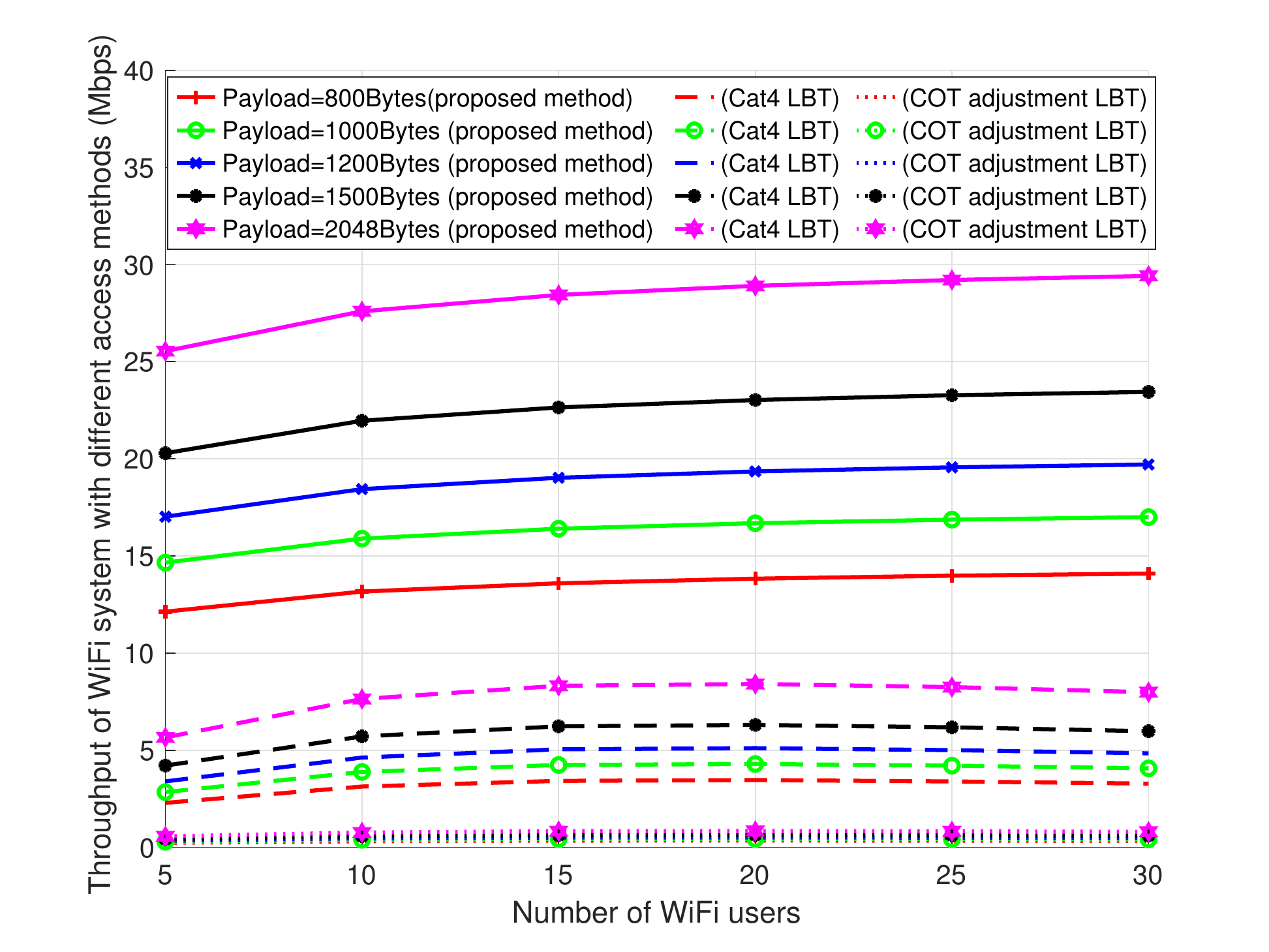}
    \label{fig:fair_method}}
\caption{Throughput fairness and WiFi throughput with different methods.} 
 \label{fig:fairness_powerpload}
\end{figure*}

\subsection{Throughput Fairness and WiFi Throughput with Different Methods}
In Fig.~(\ref{fig:fair_power}), Fig.~(\ref{fig:fair_payload}), Fig.~(\ref{fig:fair_mcot}) we compare and validate the fairness under different maximum downlink power, WiFi payload, and MCOT. It can be easily that our proposed method can obtain throughput fairness under different power, WiFi payload, MOCT. It is also observed that the maximum downlink power and MCOT have little influence on the throughput of the WiFi system, whether coexisting with a gNB system or a virtual WiFi system, while they are deeply influenced by the WiFi payload. Since gNB obtains a small successful access probability in our proposed method, and thus, the successful and failed access time $MCOT+T_{gNB}$ have little impact on the WiFi throughput coexisted with NR (under gNB) or coexisted with a virtual WiFi network (under WiFi). The maximum downlink power has no influence on WiFi throughput under gNB  and has little impact on the WiFi system under a virtual WiFi system. From Fig.~(\ref{fig:fair_payload}), we can find that the WiFi throughput under gNB and under WiFi will increase with the increasing of WiFi payload, and the fairness can be always satisfied with different payloads. In Fig.~(\ref{fig:fair_method}),  we compare the WiFi throughput coexisted with NR with Cat4 LBT \cite{Pei/2018TVT} and COT adjustment LBT \cite{wang/TVT19}. It is observed that the WiFi throughput coexisted with NR in the proposed method is larger than that of Cat4 LBT and COT adjustment LBT whatever the payload is.

\subsection{NR Throughput under Affected Parameters with Different Methods }
In this simulation, we compare the NR throughput by the proposed method with equal time allocation and equal power allocation (ETEP)  \cite{wcnc/SunSMM17}, equal time allocation and optimal power allocation (ETOP) \cite{wcnc/SunSMM17}, and optimal time allocation and equal power allocation (OTEP).
% in terms of access probability, and collision probability, airtime ratio, NR-U throughput.

\subsubsection{Impact of  Maximum Downlink Power} Fig.~(\ref{fig:nr_power}) compares the total NR throughput of the proposed method with ETEP, ETOP, OTEP methods under different maximum downlink power of gNB. It is easily observed that the proposed method can achieve a larger throughput than other methods. Furthermore,  the larger the maximum downlink power is, the larger the total NR throughput is. The methods, OTEP and ETEP, achieve the lower throughput than our proposed method and ETOP, since these two methods adopt average equal power $p_{dk}=\frac{1}{2}P_{d,k}^{max}$. It means that power allocation has more influence on NR throughout than time allocation.

\subsubsection{Impact of the WiFi Payloads}

Fig.~(\ref{fig:nr_payload}) shows the NR throughput with different methods and different WiFi payloads, where the maximum downlink power is set as $P_{d,k}^{max}=23$~dBm. We can find that our proposed method can achieve the largest throughput compared to ETOP, OTEP, and ETEP. Furthermore, the larger payload of the WiFi system will yield slightly lower NR throughput, since more WiFi payload means WiFi will occupy the unlicensed channel for a longer time to transmit data, and gNB will be given less time to transmit on the unlicensed channel. Besides, the throughput of NR decreases with the increase of the number of WiFi nodes, since the successful access probability of gNB will become small when there are more WiFi nodes as shown in Fig.~(\ref{fig:access}).

\subsubsection{Impact of Length of the MCOT}
% In this simulation, we set $MCOT=8~ms,10~ms$, the  payload is set to 1200 Bytes. 
From Fig.~(\ref{fig:nr_mcot}), we can find that the total NR throughput of the proposed method is larger than that of other methods. The method ETEP achieves the lowest NR throughput since it adopts the average power and time allocation, which cannot guarantee the maximum NR throughput. The throughput of OTEP and ETOP is larger than that of ETEP, which means that the influence of time allocation on the NR throughput is smaller than that of power allocation. Besides, it is also observed that the larger the MCOT is, the larger the total NR throughput is, as expected. Since more time for uplink and downlink transmission is the larger throughput for the NR system.

\begin{figure*}[!tb]
\centering
% \begin{tabular}{cc}
\subfloat[Different maximum downlink power.]{
% \begin{subfigure}{.3\textwidth}
% \begin{minipage}[t]{0.3\linewidth}
    %  \centering
    \includegraphics[width=0.32\textwidth]{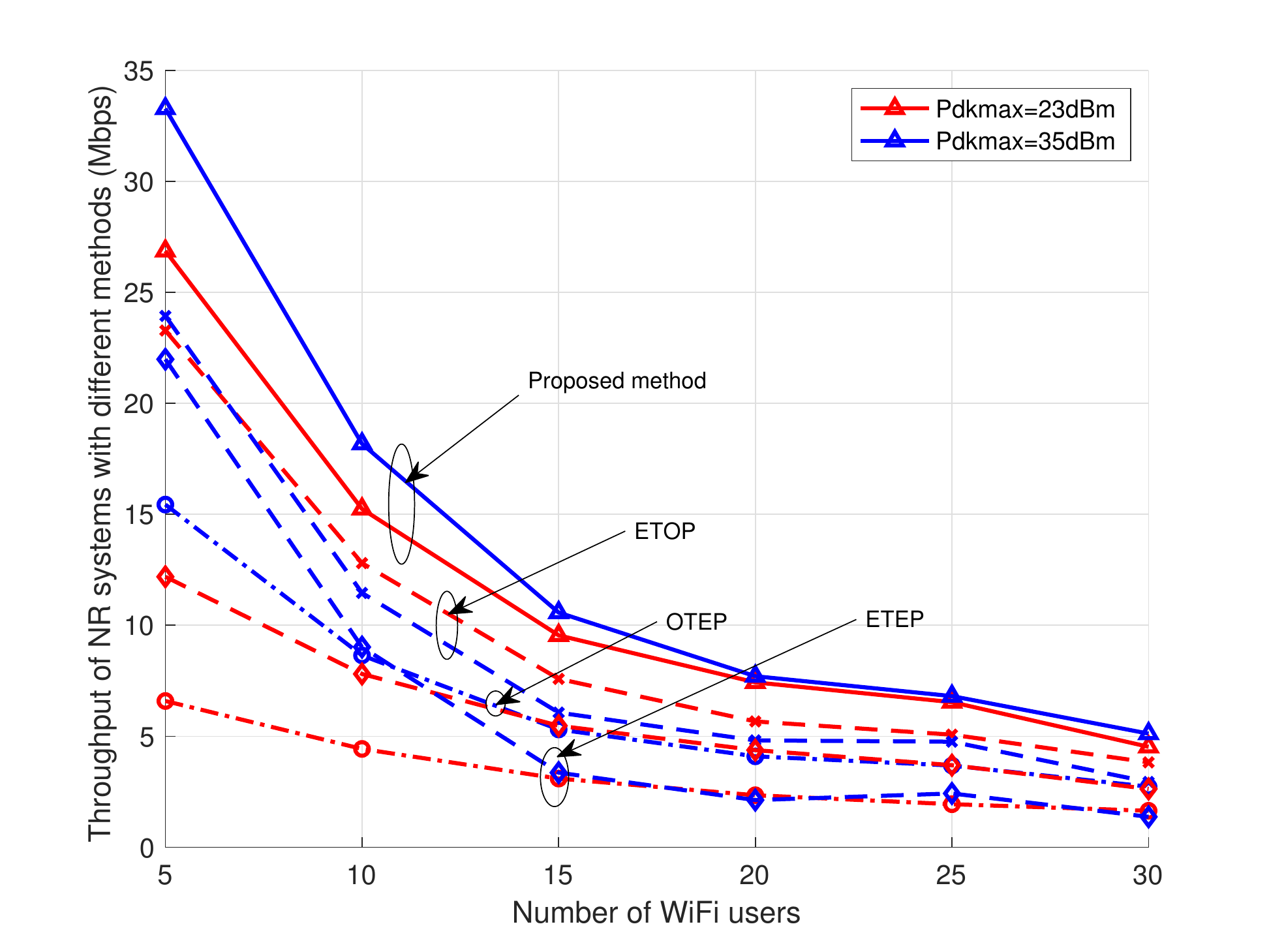}
    %  \caption{Throughput of NR system  with different power.}
    \label{fig:nr_power}
}
% \end{minipage}}
 \subfloat[Different WiFi payload.]{
% \begin{minipage}[t]{0.3\linewidth}
%  \centering
    \includegraphics[width=0.32\textwidth]{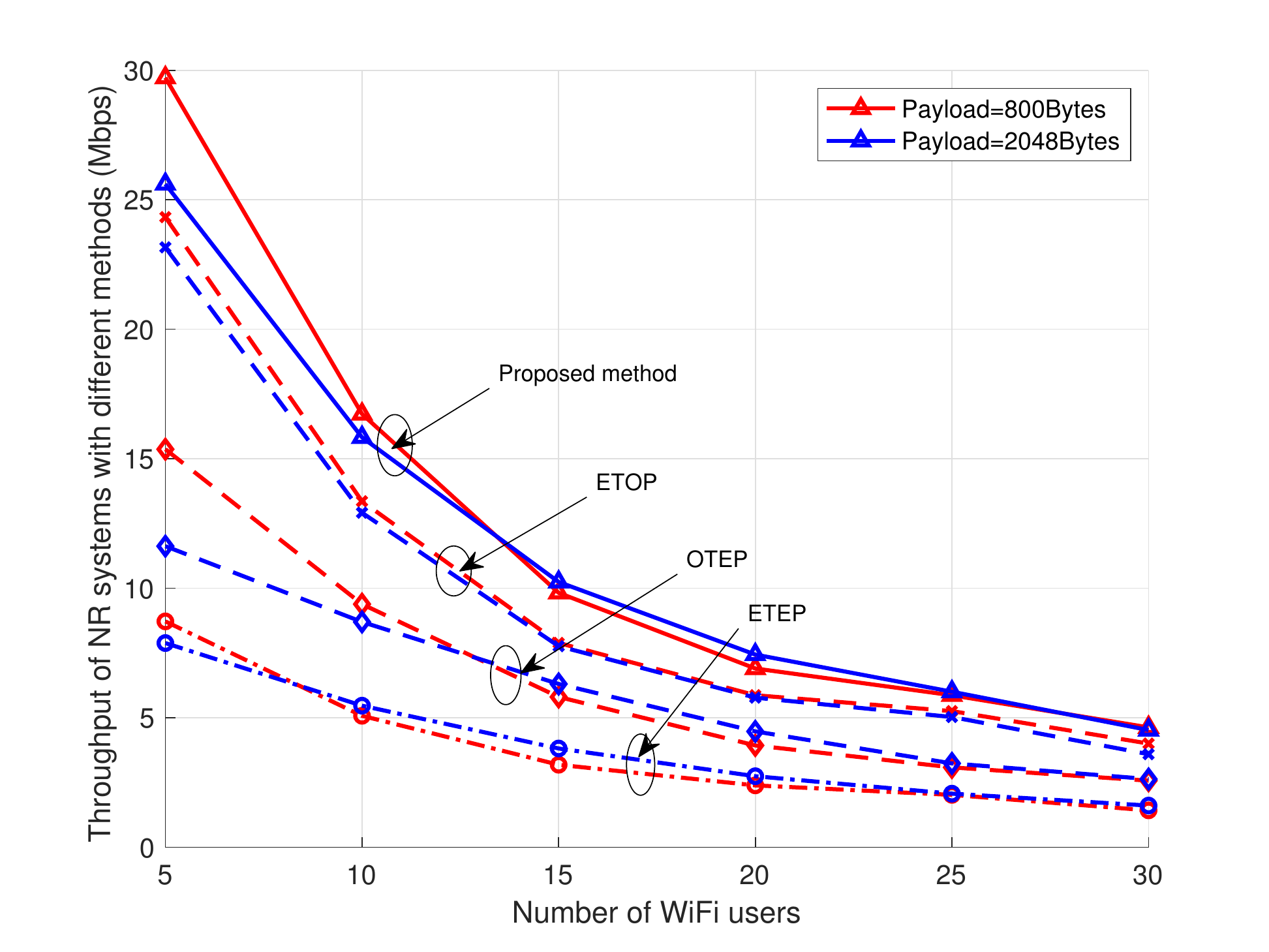}
%   \caption{Fairness under different payloads.}
     \label{fig:nr_payload}
 } 
% \end{minipage}}
\subfloat[Different MCOT.]{
% \begin{minipage}[t]{0.3\linewidth}
%  \centering
    \includegraphics[width=0.32\textwidth]{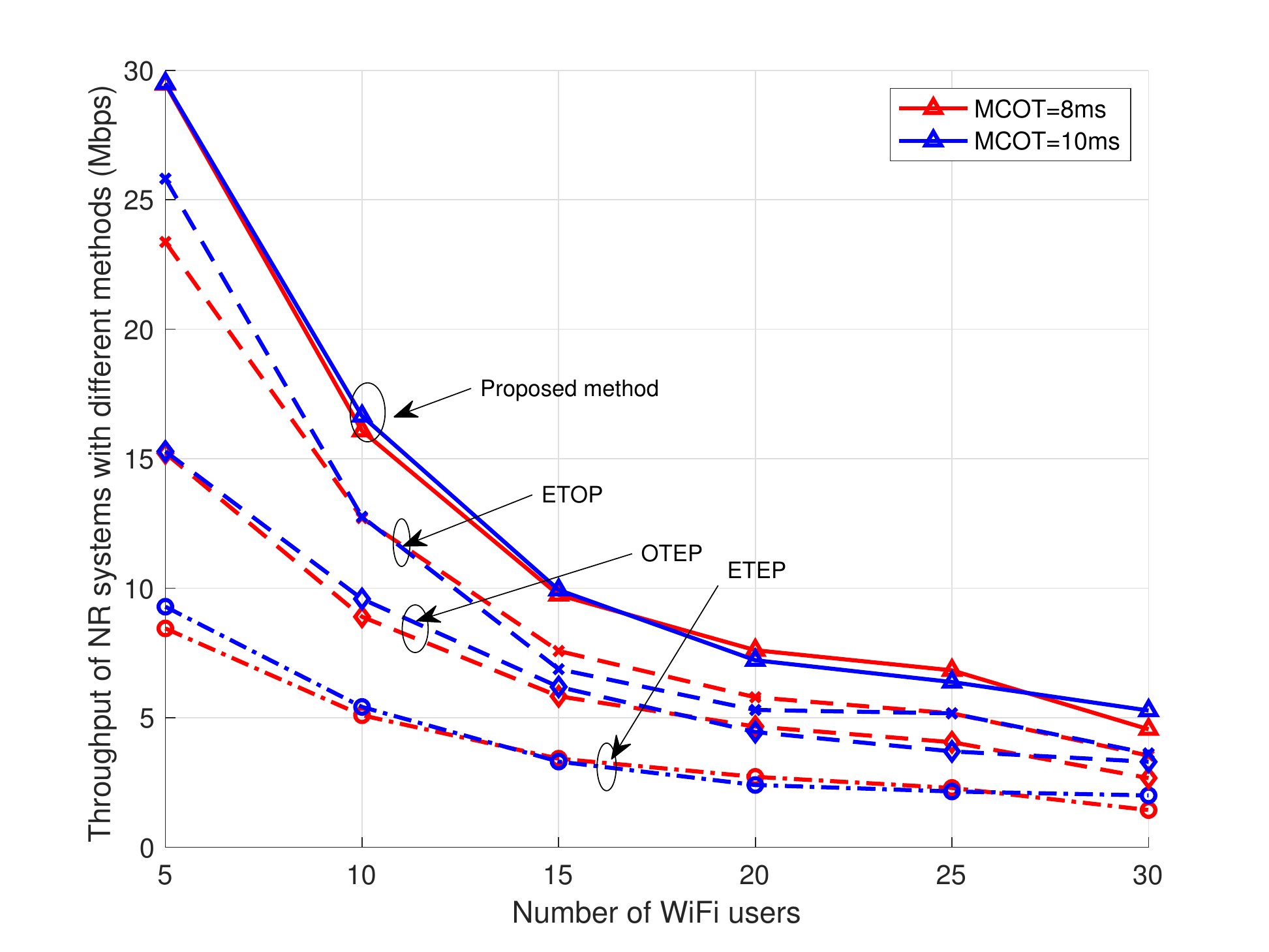}
    %  \caption{Throughput of NR system with different MCOT.}
     \label{fig:nr_mcot}
     }
% \end{minipage}}
%  \end{tabular}
\caption{Throughput of NR system under affected parameters with different methods.} 
 \label{fig:nr_compare}
\end{figure*}

\section{Conclusion}
\label{sec:conclusion}
% {\color{red} Introduction and conclusion should better highlight the distinctive contributions of the paper. (reviewer 3)}
 In this paper, we have considered coexistence between NR-U and WiFi systems under 7 GHz, and proposed a coexistence model on unlicensed channels where the MCOT of gNB is divided into two parts, one for uplink transmission and the other for downlink transmission. Our proposed equal airtime access method can make WiFi nodes and gNB obtain fair access opportunities. Furthermore, the proposed method can realize throughput fairness with different WiFi payloads, maximum downlink power and MCOTs, and achieve the largest WiFi throughput under gNB compared to Cat4 LBT and COT adjustment LBT. The optimization of time and power allocation has demonstrated superior performance over the method ETEP, ETOP, OTEP. We also find that the larger maximum downlink power, smaller WiFi payload, and larger MCOT can improve the NR throughput. 

For 5G use of unlicensed bands at higher frequencies, NR-U will need to coexist with WiGig at mmWave. The current unlicensed channel access method for mmWave unlicensed coexistence may not work well, because beamforming is necessary for directional transmission in high path loss mmWave channels, which increases the chance of coexistence in a spatial domain. In this case, directional LBT should be considered at the transmitter, or hybrid with omni-directional LBT. There are a number of studies in the literature on the coexistence of NR-U and 802.11ad. The receiver-assisted LBT, that is, listen before received (LBR), is an auxiliary method to improve the access performance, and the combination of transmitter LBT and receiver LBR  usually can provide significant enhancements in interference management. All these potential unlicensed access methods provide a path for our future research on NR-U and WiGig coexistence. We will extend the proposed methodology to the coexistence study at mmWave bands as future work.

\section{Appendix}
\label{sec:appendix}
\appendices% 
\section{Convexity proof of the formulation}
\begin{proof}
It is  well known that the Shannon formula $f(q_{d,k})=\log_2(1+\frac{MCOT \cdot q_{d,k}|{h}_{d,k}|^2}{\sigma^2})$ is concave with respect to (w.r.t.)  $q_{d,k}$, 
and $t_{d,k}f(q_{d,k})=t_{d,k}\log_2(1+\frac{MCOT \cdot q_{d,k}|{h}_{d,k}|^2}{t_{d,k}\sigma^2})$ is also concave w.r.t $(t_{d,k},q_{d,k})$ as stated in \cite{boyd/2004convex}. Thus, the first part of the objective function ($R_{d,k}=p_k t_{d,k}f(q_{d,k})$) in (\ref{equ:P2}) is also concave as the coefficient of the sum is positive w.r.t. $(\bm{t_{k}}, \bm{q})$. Similarly, we can find that 
$t_{u,k}\log_2(1+\frac{MCOT \cdot q_{u,k}|{h}_{u,k}|^2}{\sigma^2t_{u,k}})$ is also concave w.r.t. $(t_{u,k},q_{u,k})$, and thus the second part of the objective function, $R_{u,k}=p_k t_{u,k} f(q_{u,k})$, is also a concave function. Since the objective function is a  concave function  and all the constraints are affine, the maximum-concave problem is a convex problem.
\end{proof}
\section{Proof of maximum value for the second constraint}
\begin{proof}
Define
\begin{equation}
U_{i,k}=\left\{
\begin{array}{rcl}
  B_k p_k  \log_2(1+\frac{MCOT \cdot q_{d,k}|{h}_{d,k}|^2}{\sigma^2 t_{d,k}}) &,i \in 
  \mathcal{D}.\\
  B_k p_k \log_2(1+\frac{MCOT \cdot q_{u,k}|{h}_{u,k}|^2}{\sigma^2 t_{u,k}}) &, i \in 
  \mathcal{U}.\\
\end{array} \right.
\end{equation}
Thus the objective function can rewritten as 
\begin{equation}
\label{equ:equal-constraint}
\begin{split}
&\sum_{k \in \mathcal{K}} ({\sum_{i \in \mathcal{D}}t_{i,k}U_{i,k} +\sum_{i \in \mathcal{U}} t_{i,k}U_{i,k}})=\sum_{k \in \mathcal{K}} \sum_{i \in \mathcal{U}\cup \mathcal{D}} t_{i,k}U_{i,k}.
\end{split}
\end{equation}
As can be seen from (\ref{equ:equal-constraint}), the objective function is an increasing function with $t_{i,k},i \in \mathcal{U} \cup \mathcal{D}$. To maximize the objective function, the equation of the second constraint of problem (\ref{equ:P2}) should be held, i.e.,
\begin{equation}
\label{equ:equto1}
         \sum_{i \in \mathcal{D} \cup {U}} t_{i,k} =\sum_{d \in \mathcal{D}}t_{d,k}+\sum_{u \in \mathcal{U}}t_{u,k} = MCOT, \forall k \in \mathcal{K}.
\end{equation}
\end{proof}

\bibliographystyle{IEEEtran}
\bibliography{IEEEabrv,FINAL}

\begin{IEEEbiography}[{\includegraphics[width=1in,height=1.25in,clip,keepaspectratio]{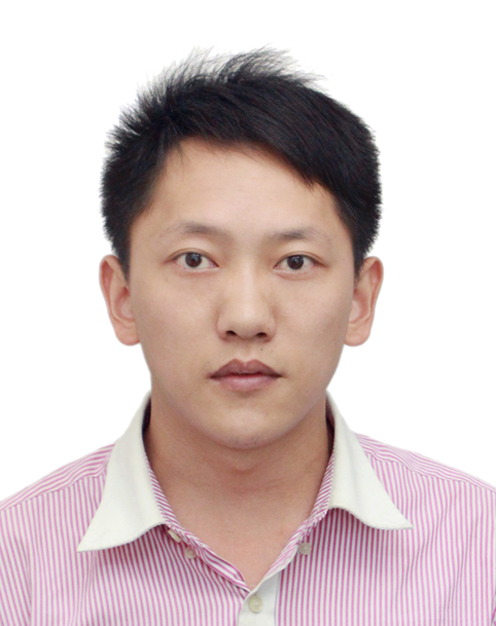}}]{Haizhou Bao}
received his B.S. degree in electronic information engineering from Huanggang Normal University in 2012, he is currently a Ph.D. student at the School of Computer Science, Wuhan University. His research interests include optimization for wireless networks, optimization of vehicular networks, data dissemination of vehicular networks, resource allocation of NR-V2X.
\end{IEEEbiography}

\begin{IEEEbiography}[{\includegraphics[width=1in,height=1.25in,clip,keepaspectratio]{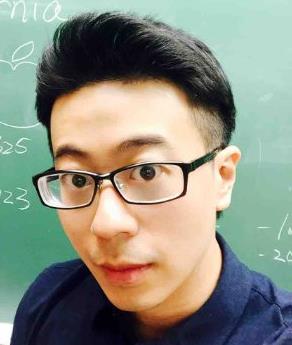}}]{Yiming Huo}
(S’08–M’18) received his B.Eng degree in information  engineering  from  Southeast  University, China,  in  2006,  and  M.Sc.  degree in System-on-Chip (SoC) from Lund University, Sweden, in 2010, and Ph.D. in electrical engineering at University of Victoria, Canada, in 2017, and he is currently a Research Associate with the same department. His recent research interests include 5G and 6G wireless systems, terahertz technology, space technology, Internet of Things, and machine learning. 

He has worked in several companies and institute including Ericsson, ST-Ericsson, Chinese Academy of Sciences, STMicroelectronics, and Apple Inc., Cupertino, CA, USA. He is a member of several IEEE societies, and also a member of the Massive MIMO Working Group of the IEEE Beyond 5G Roadmap. He was a recipient of the Best Student Paper Award of the 2016 IEEE ICUWB, the Excellent Student Paper Award of the 2014 IEEE ICSICT, and the Bronze Leaf Certificate of the 2010 IEEE PrimeAsia. He also received the ISSCC-STGA Award from the IEEE Solid-State Circuits Society (SSCS), in 2017. He has served as the Program Committee of the IEEE ICUWB 2017, the TPC of the IEEE VTC 2018/2019/2020, the IEEE ICC 2019, the Session Chair of the IEEE 5G World Forum 2018, the Publication Chair of the IEEE PACRIM 2019, the Technical Reviewer for multiple premier IEEE conferences and journals. He is an Associate Editor for the IEEE Access.
\end{IEEEbiography}

\begin{IEEEbiography}[{\includegraphics[width=1in,height=1.25in,clip,keepaspectratio]{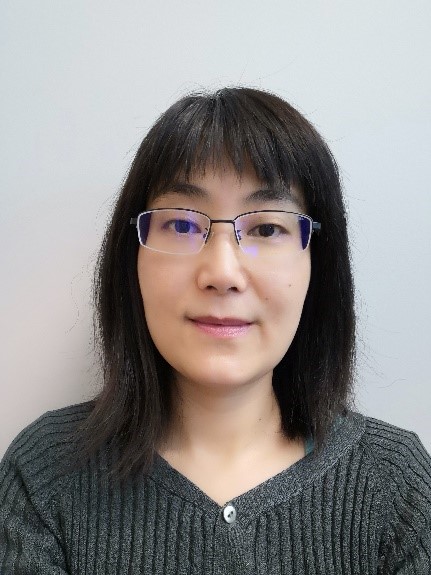}}]{Xiaodai Dong}
(S’97–M’00–SM’09) received her B.Sc. degree in Information and Control Engineering from Xi'an Jiaotong University, China in 1992, her M.Sc. degree in Electrical Engineering from National University of Singapore in 1995 and her Ph.D. degree in Electrical and Computer Engineering from Queen's University, Kingston, ON, Canada in 2000. Since January 2005 she has been with the University of Victoria, Victoria, Canada, where she is now a Professor at the Department of Electrical and Computer Engineering. She was a Canada Research Chair (Tier II) in 2005-2015. Between 2002 and 2004, she was an Assistant Professor at the Department of Electrical and Computer Engineering, University of Alberta, Edmonton, AB, Canada. From 1999 to 2002, she was with Nortel Networks, Ottawa, ON, Canada and worked on the base transceiver design of the third-generation (3G) mobile communication systems. 

Dr. Dong’s research interests include 5G, mmWave communications, radio propagation, Internet of Things, machine learning, terahertz communications, localization, wireless security, e-health, smart grid, and nano-communications. She served as an Editor for \textit{IEEE Transactions on Wireless Communications} in 2009-2014, \textit{IEEE Transactions on Communications} in 2001-2007, \textit{Journal of Communications and Networks} in 2006-2015, and is currently an Editor for \textit{IEEE Transactions on Vehicular Technology and IEEE Open Journal of the Communications Society}. 

\end{IEEEbiography}

\begin{IEEEbiography}[{\includegraphics[width=1in,height=1.25in,clip,keepaspectratio]{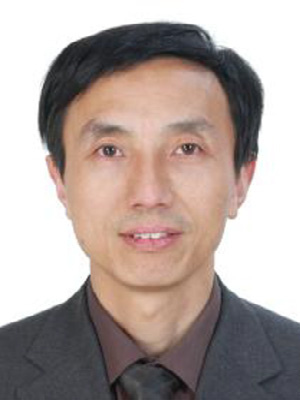}}]{Chuanhe Huang}
received his B.Sc., M.Sc., and Ph.D. degrees, all in Computer Science, from Wuhan University, Wuhan, China, in 1985, 1988, and 2002, respectively. He is currently a Professor at the School of Computer Science, Wuhan University. His research interests include computer networks, VANETs, Internet of Things, and Distributed Computing.
\end{IEEEbiography}

% -----------------------------------------------
%                     SECTION  
%\section*{Acknowledgment}


\ifCLASSOPTIONcaptionsoff
  \newpage
\fi

\end{document}